\documentclass[sigplan,screen,nonacm=true]{acmart}

\usepackage{amsmath,amsfonts, amsthm}
\usepackage{algorithmic}
\usepackage{graphicx}
\usepackage{textcomp}
\usepackage{xcolor}
\usepackage{tikz}
\usepackage{sidecap}
\usepackage{colortbl}
\usepackage{cleveref}
\usepackage{enumitem}
\usepackage{subcaption}
\usepackage{multirow}
\usepackage{xfrac}
\usepackage[linesnumbered,ruled,vlined]{algorithm2e}

\usetikzlibrary{calc}
\usetikzlibrary{backgrounds}

\newcommand{\marginparN}[1]{\marginpar{}}
\newcommand{\marginparX}[1]{}

\renewcommand{\marginpar}[1]{}

\newcommand{\mat}[1]{\ensuremath{\mathbf{#1}}}

\newcommand{\nnz}[1]{\textit{nnz}(#1)}

\newcommand{\ceil}[1]{\lceil #1 \rceil}

\newcommand{\minisec}[1]{\paragraph{#1}}
\newtheorem{lemma}{Lemma}
\newtheorem{theorem}{Theorem}
\newtheorem{corollary}{Corollary}

\begin{document}


\title[Arrow Matrix Decomposition]{Arrow Matrix Decomposition: A Novel Approach for Communication-Efficient Sparse Matrix Multiplication
}

\author{Lukas Gianinazzi}
\orcid{0000-0001-5975-4526}
\affiliation{%
  \department{Department of Computer Science}
  \institution{ETH Zurich}
}

\author{Alexandros Nikolaos Ziogas}
\orcid{0000-0002-4328-9751}
\affiliation{%
  \department{Department of Electrical Engineering}
  \institution{ETH Zurich}
}

\author{Langwen Huang}
\orcid{0000-0002-9204-0346}
\affiliation{
  \department{Department of Computer Science}
  \institution{ETH Zurich}
}

\author{Piotr Luczynski}
\orcid{0000-0002-8779-4223}
\affiliation{
  \department{Department of Computer Science}
  \institution{ETH Zurich}
}

\author{Saleh Ashkboos}
\orcid{0000-0001-6115-6779}
\affiliation{
  \department{Department of Computer Science}
  \institution{ETH Zurich}
}

\author{Florian Scheidl}
\orcid{0009-0000-5766-894X}
\affiliation{
  \department{Department of Computer Science}
  \institution{ETH Zurich}
}

\author{Armon Carigiet}
\orcid{0009-0002-3555-767X}
\affiliation{
  \department{Department of Computer Science}
  \institution{ETH Zurich}
}

\author{Chio Ge}
\orcid{0009-0005-3735-6160}
\affiliation{
  \department{Department of Computer Science}
  \institution{ETH Zurich}
}

\author{Nabil Abubaker}
\orcid{0000-0002-5060-3059}
\affiliation{
  \department{Department of Computer Science}
  \institution{ETH Zurich}
}

\author{Maciej Besta}
\orcid{0000-0002-6550-7916}
\affiliation{
  \department{Department of Computer Science}
  \institution{ETH Zurich}
}

\author{Tal Ben-Nun}
\orcid{0000-0002-3657-6568}
\affiliation{
  \department{Department of Computer Science}
  \institution{ETH Zurich}
}

\author{Torsten Hoefler}
\orcid{0000-0002-1333-9797}
\affiliation{
  \department{Department of Computer Science}
  \institution{ETH Zurich}
}


\renewcommand{\shortauthors}{Gianinazzi et al.}

\begin{abstract}
We propose a novel approach to iterated sparse matrix dense matrix multiplication, a fundamental computational kernel in scientific computing and graph neural network training.  In cases where matrix sizes exceed the memory of a single compute node, data transfer becomes a bottleneck. An approach based on dense matrix multiplication algorithms leads to suboptimal scalability and fails to exploit the sparsity in the problem. To address these challenges, we propose decomposing the sparse matrix into a small number of highly structured matrices called \emph{arrow} matrices, which are connected by permutations. Our approach enables communication-avoiding multiplications, achieving a polynomial reduction in communication volume per iteration for matrices corresponding to planar graphs and other minor-excluded families of graphs. Our evaluation demonstrates that our approach outperforms a state-of-the-art method for sparse matrix multiplication on matrices with hundreds of millions of rows, offering near-linear strong and weak scaling.
\end{abstract}

\maketitle

\begingroup
\renewcommand\thefootnote{}\footnote{\copyright\ L. Gianinazzi | ACM 2024. Author's version of the work intended for your personal use. Not for redistribution. The definitive Version of Record was published in PPoPP 2024, \url{https://dl.acm.org/doi/10.1145/3627535.3638496}.}
\addtocounter{footnote}{-1}
\endgroup

\section{Introduction}
\label{sec:intro}

Iterated sparse-dense matrix multiplications (SpMM) have numerous applications, including the training and inference of graph neural networks~\cite{DBLP:conf/sc/TripathyYB20} and the computation of eigenvectors~\cite{golub2013matrix,lanczos50Eigenvalue,Paige72Lanczos}. As the matrices arising from these problems are often too large to fit into the memory of a single GPU, they need to be \emph{decomposed} and solved on compute clusters~\cite{DBLP:journals/pc/BorstnikVWH14, DBLP:conf/ics/BuonoPCLQLT16} or processed in several batches~\cite{DBLP:journals/tpds/ZhengMLVPB17}. Data movement becomes the crucial bottleneck for such sparse workloads~\cite{DBLP:conf/ics/BuonoPCLQLT16, DBLP:conf/ics/SelvitopiBNTYB21,DBLP:journals/tpds/ZhengMLVPB17}.

Two existing approaches stand out. The first line adapts efficient algorithms designed for dense matrices to the sparse domain~\cite{DBLP:journals/pc/BorstnikVWH14, DBLP:conf/ics/SelvitopiBNTYB21, DBLP:conf/sc/TripathyYB20}. These techniques offer the advantage of low overhead and simplicity, but they encounter limitations due to their origin in dense algorithms. Consequently, their ability to fully harness the available processing power in the sparse matrix regime is limited. This deficiency forces a compromise between latency, bandwidth, and memory.

The second line of work focuses on matrix reorderings~\cite{DBLP:journals/jgt/ChinnCDG82, DBLP:journals/jcss/Feige00, DBLP:conf/swat/Feige00, DBLP:conf/acm/CuthillM69,DBLP:journals/pc/AcerSA16}. This approach involves permuting the rows and columns of a matrix to enhance computational and communication efficiency. However, these methods often rely on heuristic strategies \cite{DBLP:conf/acm/CuthillM69, DBLP:journals/siamcomp/Turner86, DBLP:conf/micro/GengWZ00YHLL21} and are constrained by unfavorable lower bounds and complexity results \cite{DBLP:conf/swat/Feige00, DBLP:journals/jcss/Feige00, DBLP:journals/computing/Papadimitriou76, DBLP:journals/tcs/Muradian03}. Of particular concern is the sensitivity of these bounds to the maximum degree and the diameter of the graph. This drawback is especially pronounced in scale-free graphs and those with skewed degree distributions.

\begin{figure*}
\begin{subfigure}[t]{0.195\textwidth}
\includegraphics[width=\textwidth]{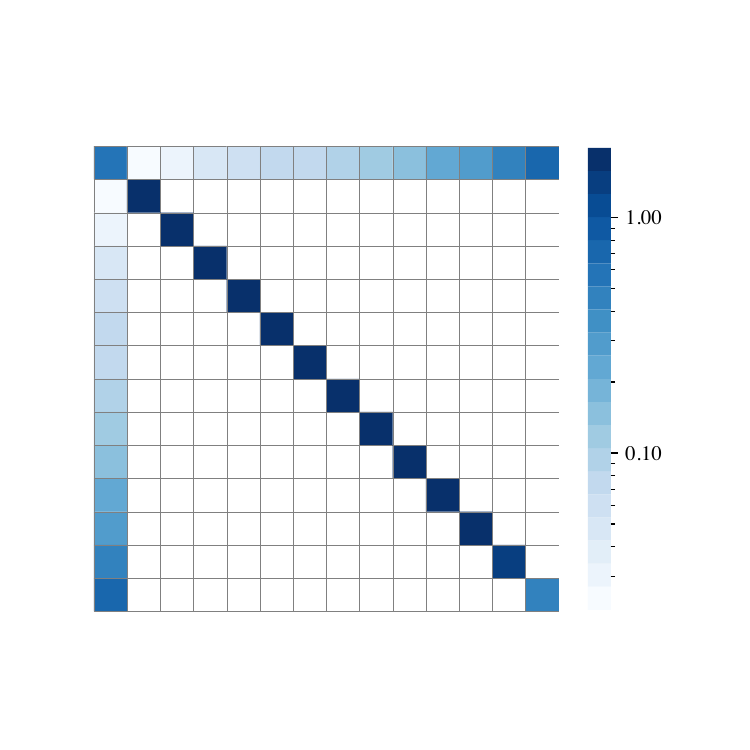}
\vspace{-3em}
\caption{GenBank 68M}
\end{subfigure}
\hfill
\begin{subfigure}[t]{0.195\textwidth}
\includegraphics[width=\textwidth]{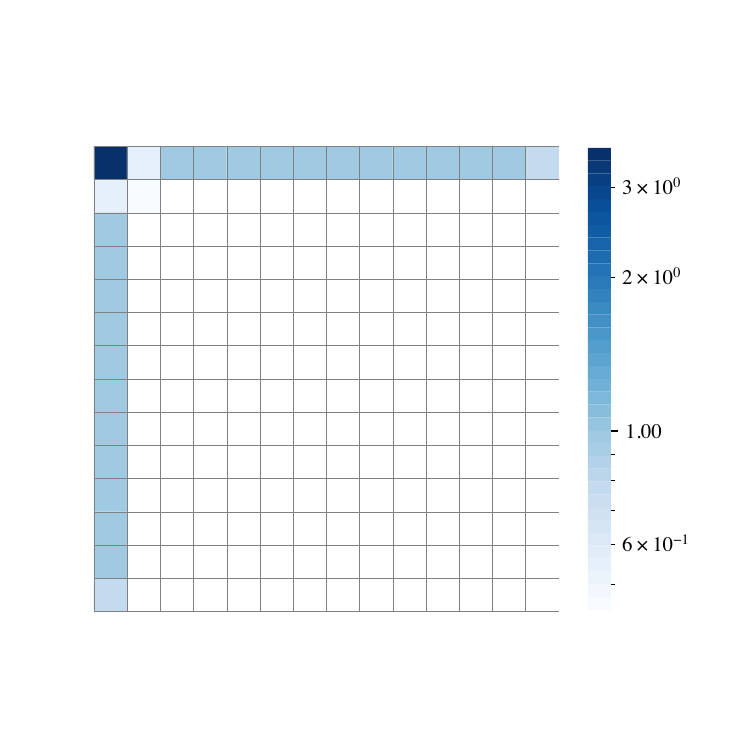}
\vspace{-3em}
\caption{MAWI 69M}
\end{subfigure}
\hfill
\begin{subfigure}[t]{0.195\textwidth}
\includegraphics[width=\textwidth]{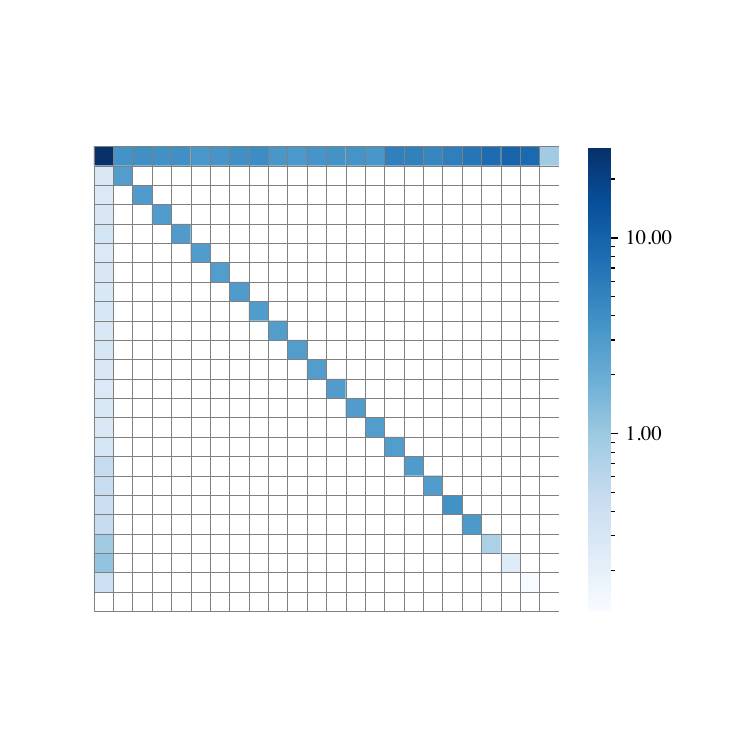}
\vspace{-3em}
\caption{Webbase 118M}
\end{subfigure}
\hfill
\begin{subfigure}[t]{0.195\textwidth}
\includegraphics[width=\textwidth]{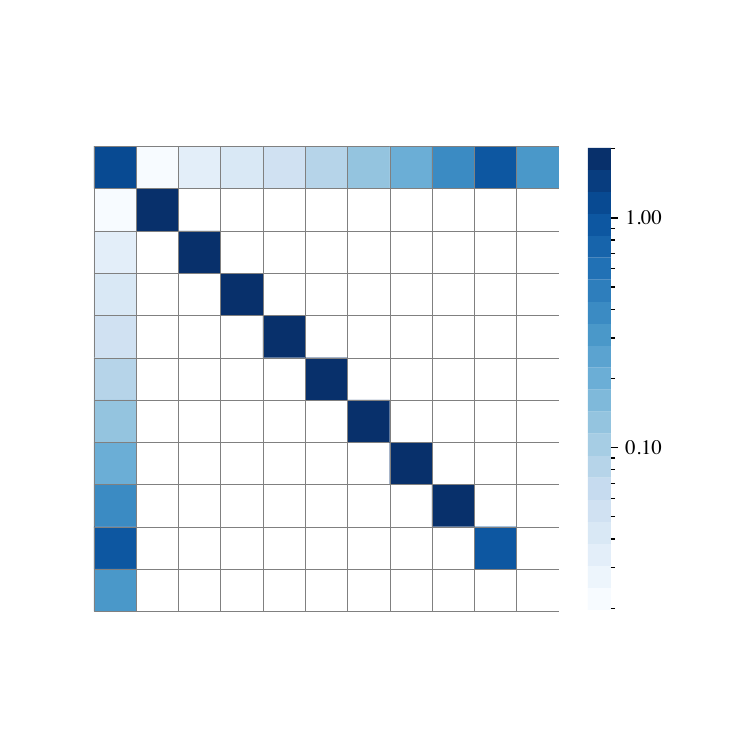}
\vspace{-3em}
\caption{OSM Europe 51M}
\end{subfigure}
\hfill
\begin{subfigure}[t]{0.195\textwidth}
\includegraphics[width=\textwidth]{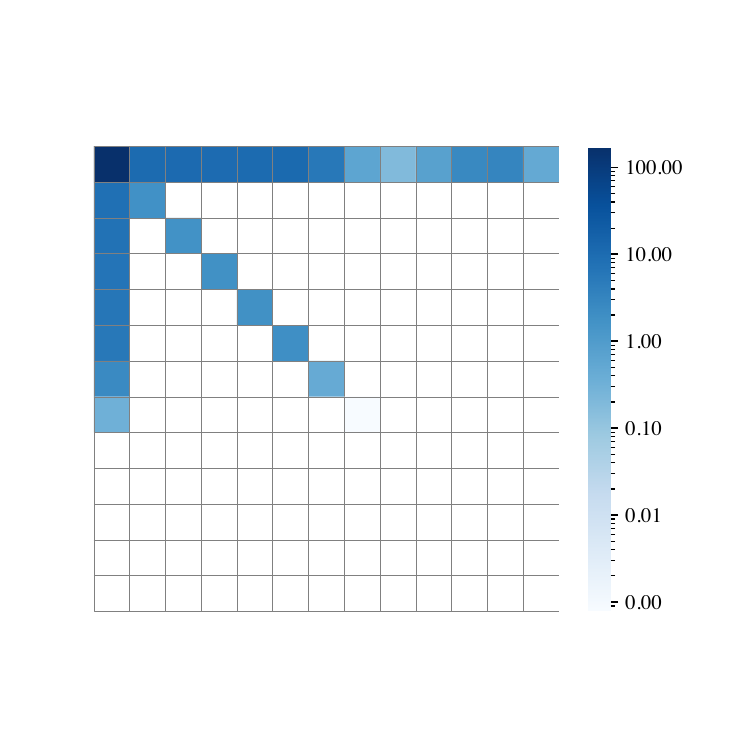}
\vspace{-3em}
\caption{GAP-twitter 62M}
\end{subfigure}
\vspace{-1em}
\caption{Non-zero structure of the first matrix \mat{B_0} in an arrow matrix decomposition for matrices from the SuiteSparse Matrix Collection. The color indicates the number of non-zeros per row; white blocks are empty. Each block has $5$ million rows. 
}
\label{fig:decomposition-illu}
\vspace{-0.5em}
\end{figure*}

To overcome these limitations, we provide a \emph{matrix decomposition} approach to sparse-matrix times dense-matrix operations. The sparse input matrix \mat{A} is \emph{decomposed} into a small number of matrices with bounded \emph{arrow-width} $b$, meaning that all non-zeros are concentrated in the first $b$ rows, columns, and a band of width $b$ around the diagonal. Formally, $\mat{A}$ has arrow-width $b$ if for all $i>b$ and $j>b$ we have that if $A_{ij}\neq 0$, then $|i-j|\leq b$. Arrow-width generalizes the notion of an arrowhead matrix~\cite{DBLP:journals/ijcm/Gravvanis98}, for which $b=1$. 
We decompose a matrix $\mat{A}$ into a sum of matrices of the form $\mat{A}=\sum_{i=1}^{l}  \mat{P}_{\pi_i} \mat{B}_i \mat{P}_{\pi_i} ^{\top}$, where each matrix $\mat{B}_i$ has arrow-width at most $b$ and each matrix $\mat{P}_{\pi_i}$ is a permutation matrix. See Figure \ref{fig:decomposition-illu} for an example of the $\mat{B_0}$ matrices.
Given this decomposition, the computation can be performed on those regularly-structured matrices in a communication-efficient way and, finally, aggregated.

In contrast to traditional bandwidth-minimization, we overcome the fundamental lower bounds with our decomposition. In particular, while any adjacency matrix of a low-diameter tree has $\Omega(n/\log n)$ bandwidth, we show, in particular, how to decompose the adjacency matrix of such a tree into $O(\log n)$ matrices of bandwidth $O(1)$. 
%
We show how to construct such an \emph{arrow matrix decomposition} for several sparsity structures, as characterized by the graphs they represent. The main idea is to use the relationship with \emph{minimum linear arrangement}~\cite{DBLP:journals/siamcomp/RaoR04, DBLP:conf/waoa/EikelSS14, DBLP:journals/algorithmica/CharikarHKR10}. 
Moreover, we prove that the pruning of high-degree vertices enabled by the arrow shape provides a \emph{polynomial} improvement in the communication volume in power law graphs.

Our proposed approach is efficient and can construct the arrow matrix decomposition in polynomial time for a variety of families of graphs, including trees, chordal graphs, planar graphs, and, more generally, $K_r$-minor free graphs. Additionally, we present a linear-time heuristic based on efficient layouts of random spanning trees that can effectively decompose real-world graphs such as biological and web traffic graphs into a small number of increasingly sparse matrices. In our evaluation, we decompose 13 of the largest matrices in the SuiteSparse matrix collection into just two to four matrices of low arrow width, whereas their maximum bandwidth can exceed $90\%$ of the number of columns.

Our approach provides significant reductions in communication costs of sparse matrix-matrix multiplication compared to traditional approaches. In our experiments, we demonstrate the scalability of our approach by testing it on several sparse matrices with over $50$ million rows.
On $128$ GPUs, our approach reduces the communication volume by $3-5$ times compared to a 1.5D decomposition, a state-of-the art approach for SpMM~\cite{DBLP:conf/ics/SelvitopiBNTYB21, DBLP:conf/sc/TripathyYB20}. 
Our approach uses less memory per compute node and effortlessly processes sparse matrices with over $200$ million vertices and dense right-hand side matrices with a larger number of columns. 

Furthermore, our evaluation shows that on a related family of sparse matrices from the same dataset, the runtime of our approach only grows by $2.3-6.2\%$ as we scale from a dataset with $18$ million rows to a dataset with over $200$ million rows when the ratio of vertices over GPUs remains constant.
Overall, our approach shows better scaling both with the number of sparse and the number of dense columns. We demonstrate good strong scaling up to $256$ compute nodes on matrices with over $200$ million rows where we show speedups of $5.3$x-$14.3$x compared to the 1.5D baseline and $1.7$x-$58$x compared to the 1D hypergraph partitioning baseline. 

%

\section{Background}
\label{sec:background}

\minisec{Graphs}
Consider an undirected graph $G$ with $n$ vertices $V(G)$ and $m$ edges $E(G)$. The subgraph of $G$ induced by $S\subseteq V$ is $G[S]$. The degree of a vertex $v$ is deg($v$) and the maximum degree of $G$ is $\Delta(G)$, or $\Delta$ for short when $G$ is clear from the context. Given a rooted tree, the set of descendants of a vertex $v$ is $v^{\downarrow}$. The diameter of a graph $G$ is $D(G)$.
%
%
If there is a permutation $\pi$ of the vertices $V(G)$ such that $\max_{(u, v) \in E(G)} |\pi(u)-\pi(v)| = w$, $G$ has \emph{bandwidth} $w$~\cite{DBLP:conf/acm/CuthillM69, DBLP:journals/jcss/Feige00}.

\minisec{Matrices} We denote the adjacency matrix of $G$ by $\mat{A}$, meaning that the number of non-zeros in $\mat{A}$ is $\nnz{\mat{A}}=m$. We consider a dense \emph{tall and skinny} \emph{feature matrix} $\mat{X}=\mat{X}_0 \in R^{n\times k}$, with $k$ \emph{features} where $k \ll n$~\cite{DBLP:conf/ics/SelvitopiBNTYB21}. Our goal is to compute the matrix iteration $\mat{X}_{t+1}=\sigma(\mat{A}\mat{X_{t}})$ for some number of steps $T$. The function $\sigma$ denotes some application-dependent element-wise function or normalization operation. We will focus our attention on the computation of the product $\mat{Y}=\mat{A}\mat{X}$ in the situation where $T\gg 1$. This means that we can afford to preprocess the problem and amortize the cost over the iterations. 
A matrix has \emph{bandwidth} $w$ if its nonzero elements are at most $w$ away from the diagonal. That is, the matrix $\mat{A}$ has bandwidth $w$ if for all $i$, $j$ we have that if $A_{ij}\neq 0$, then $|i-j|\leq w$.

\minisec{The $\alpha$-$\beta$ Model of Computation}
We consider $p$ processors that can each send and receive one message simultaneously. Sending a message of size $s$ has a latency cost $\alpha$ and a bandwidth cost $\beta\cdot s$~\cite{DBLP:journals/concurrency/ChanHPG07, DBLP:conf/ics/SelvitopiBNTYB21}. A message $M'$ depends on another message $M$ if its content, recipient, or existence depends on $M$. The latency cost of a computation is the largest sum of latency costs along a chain of dependent messages. The bandwidth cost of an algorithm is the largest total bandwidth cost over all processors.


\section{Related Work}
\label{sec:related}



Based on parallel algorithms for dense matrix multiplications~\cite{DBLP:journals/siamsc/SchatzGP16}, Selvitopi et al.~\cite{DBLP:conf/ics/SelvitopiBNTYB21} detail several approaches to tile the sparse-times-dense-skinny SpMM, which trade off communication cost with storage.
%

\minisec{1.5D $A$-stationary} 
The 1.5D $A$-stationary algorithm~\cite{ DBLP:conf/ics/SelvitopiBNTYB21,DBLP:conf/sc/TripathyYB20} arranges the processors in a $\smash{\frac{p}{c} \times c}$ grid, where $c$ is the replication factor.
It slices the matrix \mat{A} into tiles of size $\smash{\frac{nc}{p} \times \frac{n}{c}}$ (splitting it both by row and column), with each processor assigned a single tile.
It splits the feature matrix \mat{X} along the row dimension into tiles of size $\smash{\frac{nc}{p} \times k}$, meaning that each tile is replicated in the $c$ processors of a grid row.
The $\smash{\frac{p}{c}}$ processor of a grid column compute together a single $\smash{\frac{nc}{p} \times k}$-sized tile of the output \mat{Y}, with each processor holding a partial tile of the same size.
To do so, they require $\smash{\frac{p}{c^2}}$ tiles of \mat{X}.
The computation happens in $\smash{\frac{p}{c^2}}$ rounds, broadcasting one of those tiles along the grid column, so each processor needs to hold only one extra \mat{X} tile at any point of the execution.
After executing all rounds, each grid column performs an all-reduce operation to compute the full tile.
The communication cost is $\smash{O\bigl(\alpha \frac{p}{c^2} \log p + \beta \bigl(\frac{nk}{c} + \frac{nkc}{p}\bigr)\bigr)}$~\cite{DBLP:conf/ics/SelvitopiBNTYB21}.
The total storage cost for all processors is $O\left( m + cnk\right)$.
%
For the special case $c=1$, the algorithm is equivalent to a 1D version with communication cost $\smash{O\bigl(\alpha p \log p + \beta \bigl(\frac{nk}{\sqrt{p}} + nk\sqrt{p}\bigr)\bigr)}$ and total storage cost $O\left(m + nk\right)$.
For \emph{full replication}, $c=\sqrt{p}$, the communication cost is $\smash{O\bigl(\alpha \log p + \beta \frac{nk}{\sqrt{p}}\bigr)}$, and the total storage cost is $O\left(m + \sqrt{p}nk\right)$. 
High values of the replication factor thus reduce the communication cost but increase storage. 

\minisec{2D $A$-stationary} 
In contrast to the 1.5D algorithm, the feature matrix \mat{X} is sliced by columns as well in the 2D \mat{A}-stationary algorithm. However, this requires computing the result in $\sqrt{p}$ phases. This both reduces the size of the local SpMM operations (which leads to decreased local SpMM performance~\cite{DBLP:conf/ics/SelvitopiBNTYB21}) and leads to a higher communication cost. Overall, the approach needs to store $O(n/p)$ rows of the feature matrix per processor and has a total communication cost of $\smash{O\bigl(\alpha \sqrt{p}\log p + \beta \frac{nk\log p}{\sqrt{p}}\bigr)}$. Compared to the 1.5D algorithm with $c=\sqrt{p}$, this improves the storage by a factor of $\sqrt{p}$ but increases the latency cost by a factor of $\Theta(\sqrt{p}$) and the bandwidth cost by a factor of $\Theta(\log p)$. Previous work found $2D$ decompositions to scale less favorably compared to 1.5D algorithms in for skinny feature matrices~\cite{DBLP:conf/sc/TripathyYB20, DBLP:conf/ics/SelvitopiBNTYB21}.

\minisec{2D $Y$-stationary} 
In the case where the matrix $\mat{A}$ has fewer columns than rows, an algorithm that keeps the result matrix $\mat{Y}$-stationary improves communication costs~\cite{DBLP:conf/ics/SelvitopiBNTYB21}. As we focus on the case where $\mat{A}$ is square (due to being an adjacency matrix of a graph), this approach does not provide any communication cost benefits in our settings. 


\minisec{Square Dense Matrices}
For the case where the feature matrix is also square, Koanantakool et al.~\cite{DBLP:conf/ipps/KoanantakoolABM16} evaluate several communication avoiding decomposition schemes. In our work, we focus on approaches specifically designed for the case where the feature matrix is skinny, but tall. 

\minisec{Graph Partitioning} 
Several (hyper-)graph-partitioning approaches have been proposed for sparse matrix-vector~\cite{DBLP:conf/ics/SunVN17, DBLP:conf/sc/BomanDR13, page2021scalability, DBLP:journals/tpds/CatalyurekA99} and sparse matrix-matrix~\cite{DBLP:journals/topc/BallardDKS16, DBLP:conf/ipps/DevineBHBC06} multiplication. Our matrix decomposition approach is not based on graph partitioning. Instead, our objective function is based on minimum linear arrangement~\cite{DBLP:journals/algorithmica/CharikarHKR10}. Our approach avoids communication load imbalance between the partitions as we use the 1.5D $A$-stationary algorithm, specialized to arrow matrices. 

\minisec{Graph Reordering} 
If the matrix $\mat{A}$ has low bandwidth, one can efficiently compute the sparse matrix-matrix product $\mat{A}\mat{X}$ with a small communication cost and low storage requirements using a 1.5D $\mat{A}$-stationary algorithm. The bandwidth of $G$ is at least $\ceil{\frac{n-1}{D(G)}}$ and at least $\left \lceil \frac{\Delta(G)}{2} \right \rceil$~\cite{DBLP:journals/jgt/ChinnCDG82} meaning that low-diameter networks~\cite{Albert11Diameter} and power-law networks~\cite{RevModPhys.74.47} have high bandwidth. Computing the bandwidth is NP-hard~\cite{DBLP:conf/swat/Feige00}, even on bounded degree trees~\cite{DBLP:journals/jgt/ChinnCDG82}. 





\section{Arrow Matrix Decomposition}
\label{sec:decompose}

Our approach is to \emph{decompose} the graph $G$ into as few graphs as possible, each having low \emph{arrow-width}. Then, we  efficiently perform an SpMM on each of the arrow matrices and only need to aggregate the partial results. 
We demonstrate in \Cref{sec:pruning} that the arrow shape is necessary to effectively represent graphs with skewed degree distributions. 

In terms of matrices, this results in an \emph{arrow matrix decomposition}  of the form $\mat{A}=\sum_{i=1}^{l}  \mat{P}_{\pi_i} \mat{B}_i \mat{P}_{\pi_i} ^{\top}$, where each matrix $\mat{B}_i$ has arrow-width at most $b$ and each matrix $\mat{P}_{\pi_i}$ is a permutation matrix corresponding to a permutation $\pi_i$ of the vertices of the graph. We call such a decomposition a $b$-arrow matrix decomposition of order $l$. Then, we can compute $\mat{Y}=\mat{A}\mat{X}$ as
\begin{align}
	\mat{A}\mat{X} &= \sum_{i=1}^{l}  \mat{P}_{\pi_i} \left (\mat{B}_i ( \mat{P}_{\pi_i} ^{\top}  \mat{X} )\right ) \enspace , \label{eq:decomp}
\end{align}
meaning that we have reduced the computation of the product onto a series of arrow matrix multiplies, permutations, and reductions.



It is desirable for an arrow matrix decomposition that the number of non-zero rows decreases quickly with $i$, as this reduces storage and communication costs. 
Note that in this case, we can always collect the non-zeros at the top of the matrix. If the total number of non-zeros in $\mat{B}_{i+1}$ is at most $\frac{1}{x}$ times the total number of non-zeros in $\mat{B}_i$, then we say the arrow matrix decomposition is $x$-\emph{compacting}. For $x>1$, an $x$-compacting arrow decomposition has order $O(1+\log_x n)$. 
In our experiments, we will construct order $1-3$ decompositions.

\subsection{Distributed SpMM Algorithm}
\label{sec:spmm-algorithm}

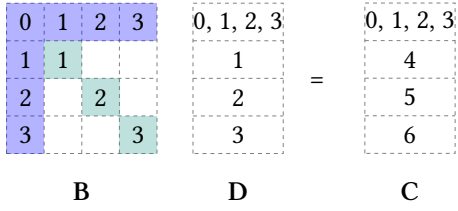
\begin{figure}
\definecolor{color2}{HTML}{2A9D8F}

\tikzstyle{matrix_elem} = [
    font=\Huge
]

  \centering
  \scalebox{0.5}{
    \setlength{\tabcolsep}{12pt}
      \begin{tabular}{ccc}
        \begin{tikzpicture}[]
        \scalebox{1}{
            \draw[help lines, dashed] grid +(4,4);
            \fill[fill=blue, fill opacity = 0.3] (0,3) rectangle (4,4);
            \fill[fill=blue, fill opacity = 0.3] (0,0) rectangle (1,3);
            \fill[fill=color2, fill opacity = 0.3] (1,2) rectangle (2,3);
            \fill[fill=color2, fill opacity = 0.3] (2,1) rectangle (3,2);
            \fill[fill=color2, fill opacity = 0.3] (3,0) rectangle (4,1);
            \node[matrix_elem] at (0.5,3.5) {0};
            \node[matrix_elem] at (1.5,3.5) {1};
            \node[matrix_elem] at (2.5,3.5) {2};
            \node[matrix_elem] at (3.5,3.5) {3};

            \node[matrix_elem] at (0.5,2.5) {1};
            \node[matrix_elem] at (0.5,1.5) {2};
            \node[matrix_elem] at (0.5,0.5) {3};

            \node[matrix_elem] at (1.5,2.5) {1};
            \node[matrix_elem] at (2.5,1.5) {2};
            \node[matrix_elem] at (3.5,0.5) {3};
        }
        \node[font=\Huge] at (2, -1) {$\mat{B}$};
        
        \end{tikzpicture}
        &
        \begin{tikzpicture}[x=0.6cm]
            \scalebox{1}{
                \draw[help lines, dashed] rectangle +(4,4);
                \draw[help lines, dashed] (0,3) -- (4,3);
                \draw[help lines, dashed] (0,2) -- (4,2);
                \draw[help lines, dashed] (0,1) -- (4,1);
                \node[matrix_elem] at (2,3.5) {0, 1, 2, 3};
                \node[matrix_elem] at (2,2.5) {1};
                \node[matrix_elem] at (2,1.5) {2};
                \node[matrix_elem] at (2,0.5) {3};

            }
            \node[matrix_elem] at (5.5, 1.9) {=};
            
            \node[font=\Huge] at (2, -1) {$\mat{D}$};
        
        \end{tikzpicture}
        &
        \begin{tikzpicture}[x=0.6cm]
            \scalebox{1}{
                \draw[help lines, dashed] rectangle +(4,4);
                \draw[help lines, dashed] (0,3) -- (4,3);
                \draw[help lines, dashed] (0,2) -- (4,2);
                \draw[help lines, dashed] (0,1) -- (4,1);
                \node[matrix_elem] at (2,3.5) {0, 1, 2, 3};
                \node[matrix_elem] at (2,2.5) {4};
                \node[matrix_elem] at (2,1.5) {5};
                \node[matrix_elem] at (2,0.5) {6};
            }
            \node[font=\Huge] at (2, -1) {$\mat{C}$};
        
        \end{tikzpicture}

      \end{tabular}
  }
\vspace{0em}
\caption{In a distribution of an arrow matrix $\mat{B}$, each tile of $\mat{B}$ is $b\times b$ and each tile of $\mat{D}$ and $\mat{C}$ is $b\times k$. The numbers indicate the process ranks holding or contributing to the tile. 
}
\vspace{0em}
\label{fig:distribution}
\end{figure}

We present a distributed algorithm for SpMM using an arrow matrix decomposition. In \Cref{sec:spmm}, we analyze the data movement and storage requirements of this algorithm in the $\alpha-\beta$ model. In particular, it improves bandwidth cost and storage requirements by a factor of $\Theta(\sqrt{p})$ at a similar latency cost compared to a fully replicated 1.5D decomposition. 

\begin{algorithm}
\caption{Arrow Matrix Multiply}\label{alg:arrow-multiply}
\DontPrintSemicolon
\KwData{\mat{B}, \mat{D}, rank $r$}
\KwResult{\mat{C}=\mat{B}\mat{D}}
Broadcast($\mat{D^{(0)}}$, root=0) \;
$\mat{C^{(0)}} = \mat{B^{(0, i)}} \mat{D^{(i)}}$ \;
Reduce($\mat{C^{(0)}}$, root=0) \;
\If{$r > 0$}{
    $\mat{C^{(r)}} = \mat{B^{(r, 0)}} \mat{D^{(0)}} +  \mat{B^{(r, r)}} \mat{D^{(r)}}$ \;
}
return $\mat{C^{(r)}}$\;
\end{algorithm}

\minisec{Arrow Matrix SpMM} Let us begin with how to compute the product $\mat{B}\mat{D}=\mat{C}$ when $\mat{B}$ has arrow width $b$. The arrow matrix's non-zeros appear in three bands, leading a 1.5D decomposition to result in most tiles of $\mat{B}$ being zero, thus yielding a communication-efficient algorithm. To further enhance efficiency, we consider a block-diagonal band. 

We tile the $n\times n$ matrix $\mat{B}$, with arrow width $b$, into $b \times b$ tiles, indexed as $\mat{B^{(i, j)}}$. Due to the arrow structure, the non-zeros occur in three types of tiles: $\ceil{\frac{n}{b}}$ row tiles ($\mat{B^{(0,j)}}$), $\ceil{\frac{n}{b}}$ column tiles ($\mat{B^{(i,0)}}$ for $i>0$), and $\ceil{\frac{n}{b}}$ diagonal tiles ($\mat{B^{(i,i)}}$). The matrices $\mat{D}$ and $\mat{C}$ are sliced into $b \times k$ tiles, indexed as $\mat{D^{(i)}}$ and $\mat{C^{(i)}}$. Each rank $i$ initially holds three tiles of $\mat{B}$ and one slice of $\mat{D}$, as depicted in Figure \ref{fig:distribution}. The multiplication using an arrow matrix is detailed in Algorithm \ref{alg:arrow-multiply} and involves two collective communication operations.

\begin{algorithm}
\caption{Arrow Decomposition Multiply}\label{alg:arrow-dec-multiply}
\DontPrintSemicolon
\KwData{Arrow Decomposition $\mat{A}=\sum_{i=1}^{l}  \mat{P}_{\pi_i} \mat{B}_i \mat{P}_{\pi_i} ^{\top}$,  $\mat{X}$, \\ rank $r$, where rank $r$ belongs to the $j$-th arrow matrix.}
\KwResult{\mat{Y}=\mat{A}\mat{X}}
\For{$k \leftarrow 1$ \KwTo $l$}{
    \If{$k == j$} {
        Send \mat{X^{(r)}} to matrix $j+1$\;
    }
    \If{$k + 1 == j$} {
        Receive $\mat{X^{(r)}}=(\mat{P}_{\pi_{j}}^{\top} \mat{X})^{(r)}$ from matrix $j-1$\;
    }
}
$\mat{Y_j^{(r)}} = (\mat{B}_j  (\mat{P}_{\pi_{j}}^{\top} \mat{X}))^{(i)} $ \;
\For{$k \leftarrow l$ \KwTo $1$}{
    \If{$k == j$} {
        Send $\mat{Y^{(r)}}$ to matrix $j-1$\;
    }
    \If{$k-1 == j$} {
        Receive $\mat{\hat Y^{(r)}}=(\sum_{i={r+1}}^{l} \mat{P}_{\pi_i} \mat{Y}_i)^{(r)}$ from $j+1$\;
        \mat{Y^{(r)}} += \mat{\hat Y^{(r)}}
    }
}
\end{algorithm}

\minisec{SpMM Algorithm}
Next, we describe how to multiply with a matrix given its arrow decomposition. Each rank is assigned to one of the matrices of the decomposition. Each matrix is distributed as in \Cref{fig:distribution}. Initially, the first matrix alone contains the input matrix \mat{X}, with each of its ranks $r$ holding a distinct block \mat{X^{(r)}}.  This block is then sent to the subsequent matrix in the sequence, propagating through each matrix. This propagation utilizes a specific permutation, $\pi_{j+1} \circ \pi^{-1}_j$, to shuffle the rows when transmitting from matrix $j$ to matrix $j+1$. Each matrix then computes its local product as described in Algorithm \ref{alg:arrow-multiply}, resulting in a partial output \mat{Y_j}, with the segment \mat{Y_j^{(j)}} stored by the rank $r$ assigned to matrix $j$. Finally, the partial results \mat{Y_j} are aggregated in reverse order, following the opposite pattern of the input matrix distribution. For a detailed explanation of this procedure, refer to Algorithm \ref{alg:arrow-dec-multiply}.

\section{Constructing the Decomposition}
\label{sec:mla}

\begin{figure*}[t]
\centering
\definecolor{color1}{HTML}{E9C46A}
\definecolor{color2}{HTML}{2A9D8F}
\definecolor{color3}{HTML}{E76F51}
\definecolor{color4}{HTML}{ffafcc}

\usetikzlibrary{backgrounds}  
\usetikzlibrary{positioning}
\tikzstyle{proc}=[
    circle,
    minimum size =0.05cm,
    draw=black,
    fill=gray,
    text=white,
    font=\small,
    inner sep=2
]

\tikzstyle{node} = [
  circle,
  draw=black,
  fill=white,
  text=black
]

\tikzstyle{nodesmall}=[
  node,
  font=\tiny,
  inner sep=2
]

\tikzstyle{child1} = [
  node,
  fill=color2
]

\tikzstyle{child2} = [
  node,
  fill=color3
]
\tikzstyle{child3} = [
  node,
  fill=color4
]

\tikzstyle{pruned} = [
    dashed,
    color=blue,
    line width=1
]

\tikzstyle{banded} = [
    color=color2,
    line width=1
]

\tikzstyle{outside} = [
    color=red,
    line width=1
]

\tikzstyle{vertex_name} = [
    font=\large
]

  \centering
  \scalebox{0.49}{
    \setlength{\tabcolsep}{12pt}
      \begin{tabular}{ccccc}
        \begin{tikzpicture}[baseline={(1,-4.5)}]
        \scalebox{1}{
            \node[proc] (3) {};
            \node[proc] (2) [left = 9mm of 3]{};
            \node[proc] (1) [left=9mm of 2]{};
            \node[proc] (4) [right=9mm of 3]{};
            \node[proc] (5) [right=9mm of 4]{};
            \node[proc] (6) [right=9mm of 5]{};
            \node[proc] (7) [right=9mm of 6]{};

            \node[vertex_name, below=8mm of 1] {${\pi^{-1}_0(0)}$};
            \node[vertex_name,  below=8mm of 2] {${\pi^{-1}_0(1)}$};
            \node[vertex_name,  below=8mm of 3] {${\pi^{-1}_0(2)}$};
            \node[vertex_name,  below=8mm of 4] {${\pi^{-1}_0(3)}$};
            \node[vertex_name,  below=8mm of 5] {${\pi^{-1}_0(4)}$};
            \node[vertex_name,  below=8mm of 6] {${\pi^{-1}_0(5)}$};
            \node[vertex_name,  below=8mm of 7] {${\pi^{-1}_0(6)}$};

            \draw[pruned] (1) -- (2);
            \draw[pruned] (1) to[out=45, in=135] (5);
            \draw[pruned] (1) to[out=45, in=135] (6);
            
            \draw[pruned] (2) to[out=45, in=135] (4);
        
            \draw[banded] (3) -- (4);
            \draw[banded] (4) -- (5);
            \draw[banded] (6) -- (7);
        
            \draw[banded] (6) -- (7);
            \draw[banded] (6) -- (7);
            
            \draw[outside] (3) to[out=315, in=225] (6);
        
            \draw[color=blue, line width=1] ($(1) + (-0.25, 0.25)$) -- ++(1.6, 0) -- ++(0, -0.5) -- ++(-1.6, 0) -- ++(0, 0.5);
            }
            \node[font=\Huge] at (1, -4) {$\pi_0$};
            
        \end{tikzpicture}
        &
        \begin{tikzpicture}[]
        \scalebox{0.94}{
            \draw[help lines, dashed] grid +(7,7);
            
            \node[vertex_name] at (0.5,7.5) {${\pi^{-1}_0(0)}$};
            \node[vertex_name] at (1.5,7.5) {${\pi^{-1}_0(1)}$};
            \node[vertex_name] at (2.5,7.5) {${\pi^{-1}_0(2)}$};
            \node[vertex_name] at (3.5,7.5) {${\pi^{-1}_0(3)}$};
            \node[vertex_name] at (4.5,7.5) {${\pi^{-1}_0(4)}$};
            \node[vertex_name] at (5.5,7.5) {${\pi^{-1}_0(5)}$};
            \node[vertex_name] at (6.5,7.5) {${\pi^{-1}_0(6)}$};
            
            \fill[fill=blue] (1,6) rectangle (2,7);
            \fill[fill=blue] (4,6) rectangle (6,7);
            \fill[fill=blue] (1,6) rectangle (0,5);
            \fill[fill=blue] (1,3) rectangle (0,1);
            
            \fill[fill=blue] (3,5) rectangle (4,6);
            \fill[fill=blue] (1,3) rectangle (2,4);

            \fill[fill=color2] (3,4) rectangle (4,5);
            \fill[fill=color2] (2,3) rectangle (3,4);
            
            \fill[fill=color2] (4,3) rectangle (5,4);
            \fill[fill=color2] (3,2) rectangle (4,3);
            
            \fill[fill=color2] (5,0) rectangle (6,1);
            \fill[fill=color2] (6,1) rectangle (7,2);

            \fill[fill=red] (5,4) rectangle (6,5);
            \fill[fill=red] (2,1) rectangle (3,2);

            \draw[color=blue, line width=1.25] (2,0) -- (2,5) -- ++(5,0) -- ++(0,2) -- ++(-7,0) -- ++(0,-7) -- ++(2,0);
            \draw[color=color2, line width=1.25] (5,5) -- ++(0,-1) -- ++(1,0) -- ++ (0,-1) -- ++(1,0) -- ++ (0,-3) -- ++ (-3,0)  -- ++(0,1) -- ++ (-1,0) -- ++(0,1) -- ++ (-1,0) ;
        }
        \node[font=\Huge] at (3.5, -1) {$\mat{P_{\pi_0}^\top} \mat{A} \mat{P_{\pi_0}}$};
        
        \end{tikzpicture}
        &
        \begin{tikzpicture}[]
        \scalebox{0.94}{
            \draw[help lines, dashed] grid +(7,7);
            
            \node[vertex_name] at (0.5,7.5) {${\pi^{-1}_0(0)}$};
            \node[vertex_name] at (1.5,7.5) {${\pi^{-1}_0(1)}$};
            \node[vertex_name] at (2.5,7.5) {${\pi^{-1}_0(2)}$};
            \node[vertex_name] at (3.5,7.5) {${\pi^{-1}_0(3)}$};
            \node[vertex_name] at (4.5,7.5) {${\pi^{-1}_0(4)}$};
            \node[vertex_name] at (5.5,7.5) {${\pi^{-1}_0(5)}$};
            \node[vertex_name] at (6.5,7.5) {${\pi^{-1}_0(6)}$};
            
            \fill[fill=blue] (1,6) rectangle (2,7);
            \fill[fill=blue] (4,6) rectangle (6,7);
            \fill[fill=blue] (1,6) rectangle (0,5);
            \fill[fill=blue] (1,3) rectangle (0,1);
            
            \fill[fill=blue] (3,5) rectangle (4,6);
            \fill[fill=blue] (1,3) rectangle (2,4);

            \fill[fill=color2] (3,4) rectangle (4,5);
            \fill[fill=color2] (2,3) rectangle (3,4);
            
            \fill[fill=color2] (4,3) rectangle (5,4);
            \fill[fill=color2] (3,2) rectangle (4,3);
            
            \fill[fill=color2] (5,0) rectangle (6,1);
            \fill[fill=color2] (6,1) rectangle (7,2);

            \draw[color=blue, line width=1.25] (2,0) -- (2,5) -- ++(5,0) -- ++(0,2) -- ++(-7,0) -- ++(0,-7) -- ++(2,0);
            \draw[color=color2, line width=1.25] (5,5) -- ++(0,-1) -- ++(1,0) -- ++ (0,-1) -- ++(1,0) -- ++ (0,-3) -- ++ (-3,0)  -- ++(0,1) -- ++ (-1,0) -- ++(0,1) -- ++ (-1,0) ;
        }
        \node[font=\Huge] at (3.5, -1) {$\mat{B_0}$};
        
        \end{tikzpicture}
        &
        \begin{tikzpicture}[baseline={(-1,-4.5)}]
        \scalebox{2}{
            \node[proc] (1) {};
            \node[proc] (2) [right of=1]{};
            \node[font=\tiny, below=3mm of 1] {${\pi^{-1}_1(0)}$};
            \node[font=\tiny, below=3mm of 2] {${\pi^{-1}_1(1)}$};
            \draw[pruned] (1) -- (2);
            \draw[color=blue, line width=1] ($(1) + (-0.25, 0.25)$) -- ++(1.5, 0) -- ++(0, -0.5) -- ++(-1.5, 0) -- ++(0, 0.5);
        }
        \node[font=\Huge] at (1, -4) {$\pi_1$};
        \end{tikzpicture}
        &
        \hspace{8mm}
        \begin{tikzpicture}[]
          \scalebox{0.94}{
              \draw[help lines, dashed] grid +(7,7);
        
              \node[vertex_name] at (0.5,7.5) {${\pi^{-1}_1(0)}$};
              \node[vertex_name] at (1.5,7.5) {${\pi^{-1}_1(1)}$};
              \node[vertex_name] at (2.5,7.5) {${\pi^{-1}_1(2)}$};
              \node[vertex_name] at (3.5,7.5) {${\pi^{-1}_1(3)}$};
              \node[vertex_name] at (4.5,7.5) {${\pi^{-1}_1(4)}$};
              \node[vertex_name] at (5.5,7.5) {${\pi^{-1}_1(5)}$};
              \node[vertex_name] at (6.5,7.5) {${\pi^{-1}_1(6)}$};
        
        
              \fill[fill=blue] (1,6) rectangle (2,7);
              \fill[fill=blue] (1,6) rectangle (0,5);
        
              \draw[color=blue, line width=1.25] (2,0) -- (2,5) -- ++(5,0) -- ++(0,2) -- ++(-7,0) -- ++(0,-7) -- ++(2,0);

          }
    
          \node[font=\Huge] at (3.5, -1) {$\mat{B_1}$};
        \end{tikzpicture}

      \end{tabular}
  }
\vspace{-0.5em}
\caption{\textsc{LA-Decompose} produces a linear arrangement $\pi_0$ of the vertices of the graph that corresponds to the sparsity structure of the input matrix. This creates three parts in the matrix (1) A flipped 'L' shape that contains the highest degree vertices (in blue), (2) a band around the diagonal (in green), and (3) the remainder (in red). The first two parts form the first matrix $\mat{B_0}$ of the decomposition. The rest of the decomposition proceeds recursively on the remainder.}
\label{fig:la-decompose}
\end{figure*}
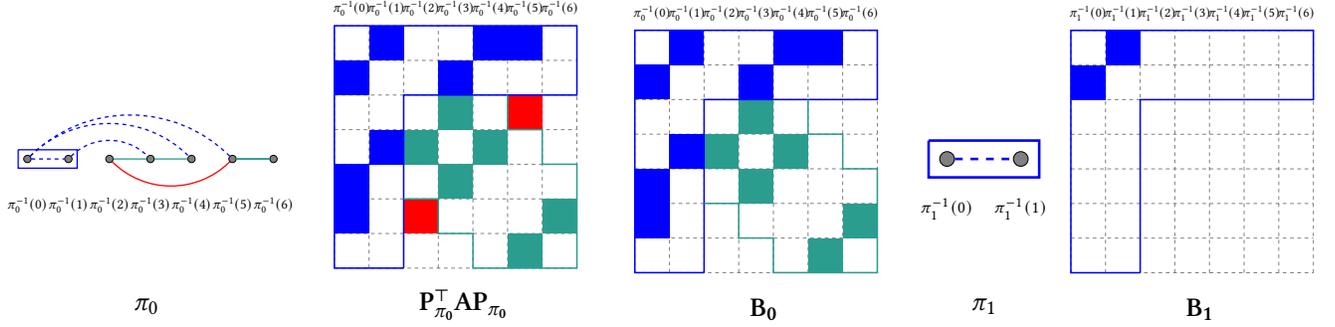

When we construct an arrow matrix decomposition, there is a trade-off between the time to compute the decomposition and its compactness. We frame the decomposition problem in a graph-theoretic language, which allows us to obtain algorithms that are polynomial time and provide strong bounds on certain sparsity structures. 

The high-level idea is to consider the matrix as a graph and find a permutation of its vertices, a so-called \emph{linear arrangement}, such that many edges connect vertices that are close in the order of the permutation. We minimize a cost function that sums over the distances of the edges in the linear arrangement. Because high-degree vertices add high costs, we collect those at the beginning of the order. Then, we construct a remainder graph consisting of the edges that are much further apart than the average cost of the solution and proceed recursively.

In addition to provable polynomial time bounds on several families of graphs, we present a near-linear time heuristic based on high-quality linear arrangements of random spanning trees. Because of its scalability to hundreds of millions of nodes, we use this random spanning forests approach to compute our decompositions in our evaluation. 

\subsection{LA-Decompose}\label{sec:la-decompose}

Any reordering of the matrices and rows of a square matrix can be viewed as a permutation of the vertices of its graph. Such a permutation is 
 called a \emph{linear arrangement}. Our goal is to find a permutation that leads to most of the non-zeros being close to the diagonal. Hence, we consider the cost function $\lambda_{\pi}(G)$ of the linear arrangement $\pi$ as $\lambda_\pi(G)=\sum_{(u, v) \in E(G)} |\pi(u)-\pi(v)|$. If the graph $G$ is clear from the context, we omit $G$ from the notation. A linear arrangement of $G$ with the smallest cost is a \emph{minimum linear arrangement} (MLA)~\cite{DBLP:journals/siamcomp/RaoR04, DBLP:conf/waoa/EikelSS14, DBLP:journals/algorithmica/CharikarHKR10}. Computing a minimum linear arrangement is NP-hard, however, it can be approximated in polynomial time within a $O(\sqrt{\log n} \log \log n)$ factor~\cite{DBLP:journals/algorithmica/CharikarHKR10} and solved exactly in polynomial time on trees~\cite{DBLP:journals/ipl/Alemany-PuigEF22} and chordal graphs~\cite{DBLP:journals/isci/RaoufiRB13}.
Note that a graph with bandwidth $b$ has a linear arrangement of cost at most $mb$ or $nb^2$. In contrast, we show there are graphs with a linear arrangement of cost $O(n \log n)$ but bandwidth $\Omega(n/\log n)$.

The idea of our algorithm is that a linear arrangement $\pi$ of cost $\lambda_\pi(G)$ concentrates a constant fraction of the non-zeros along a $O(\lambda_\pi(G)/m)$-wide band along the diagonal. Removing this portion and repeating the process until no edges are left leads to a compact arrow decomposition. We present a framework for computing an arrow matrix decomposition using a linear arrangement, called \textsc{LA-Decompose}($\mat{A}$, $b$):

We are given a matrix $\mat{A}$ and a desired arrow-width $b \geq 2$. Set $\mat{A}_0=\mat{A}$ and $i=0$. Until the number of non-zeros in $\mat{A}_i$ is at most $2b$, repeat the following steps:
\begin{enumerate}[leftmargin=0.5cm]
    \item Place the $b$ highest degree vertices $V_i^{h}$ at the beginning of the linear arrangement $\pi_i$.
	\item Compute a linear arrangement $\pi'_i$ of the induced subgraph $G_i'=G_i[V_i\backslash V_i^{h}]$ of $\mat{A}_i$ and concatenate it to $\pi_i$.
	\item Set $\mat{B}_i$ equal to the submatrix of $\mat{P}_{\pi_i}^{\top}\mat{A}_i \mat{P}_{\pi_i}$ consisting of the first $b$ rows and columns and a symmetric $b$-wide band around the diagonal.
	\item Set $\mat{A}_{i+1}=\mat{A}_i - \mat{P}_{\pi_i} \mat{B}_i \mat{P}_{\pi_i}^{\top} $ and increment $i$. 
\end{enumerate}
Observe that the matrices $\mat{A}_i$ do not need to be constructed explicitly and one can instead work on the corresponding graphs. 
See \Cref{fig:la-decompose} for an illustration of the approach.

\begin{lemma}\label{lem:la-decompose}
For $x = \frac{bm}{ \max_i \lambda_{\pi_i'}(G_i')}$, 
\textsc{LA-Decompose}($\mat{A}$, $b$) computes an $x$-compacting $b$-arrow matrix decomposition. 
\end{lemma}
\begin{proof}
The arrow-width of any of the matrices $\mat{B}_i$ is $b$ by construction. Moreover, in step $i$ the average distance from the diagonal is at most $\frac{\lambda_{\pi_i'}(G_i')}{m}$. No more than a $\frac{1}{x}$ fraction of the entries can be more than $x$-times the average away from the diagonal. Hence, in each iteration at most a $\frac{1}{x}$ fraction of the edges remain, i.e., the decomposition is $x$-compacting. 
\end{proof}
\Cref{lem:la-decompose} means that the number of nonzeros decreases geometrically for the matrices in the decomposition as long as $b$ is larger than the average cost of an edge in the linear arrangements $\pi'_i$. In our experiments, this will always be the case for our choices of $b$. 


Next, we will show efficient algorithms for linear arrangements and prove lower bounds for certain families of graphs.
The bounds on the cost of the linear arrangement will depend necessarily on the maximum degree $\Delta$ of the graph. This is why we removed the highest degree vertices before computing the linear arrangement in \textsc{LA-Decompose}. In \Cref{sec:pruning}, we show how the pruning of high-degree vertices improves the arrow decomposition in graphs with a power law degree distribution, which occur in real-word graphs~\cite{RevModPhys.74.47}.

\subsection{Linear Arrangement using Separators}\label{sec:separator-la}

We show how to construct linear arrangements efficiently by recursively partitioning the graph. This allows us to obtain bounds on the cost of a linear arrangement for several families of graphs that can be separated efficiently. 

A set of vertices $S$ whose removal leaves the graph with connected components of size at most $\frac{2}{3}n$ is a $\frac{2}{3}$-separator. 
For any positive integer $k\leq n$, let $s_k(G)$ be the smallest number such that all subgraphs of $G$ containing at most $k$ vertices have a $\frac{2}{3}$-separator of size $s_k(G)$. The \emph{separation number} $s(G)$ of $G$ is $s(G)=s_n(G)$. 
 Small separators can be constructed, in particular, for clique-minor free graphs~\cite{DBLP:conf/focs/KawarabayashiR10} and bounded treewidth graphs~\cite{DBLP:journals/ejc/BottcherPTW10, ROBERTSON1986115}.

\textsc{Separator-LA}$(G)$ constructs a linear arrangement using separators recursively:

\begin{enumerate}[leftmargin=0.5cm]
	\item Compute a $\frac{2}{3}$-separator $S$ of the current subgraph $G$.
	\item Place the vertices of $S$ at the beginning of the linear order.
	\item Then, place the connected components of $G[V(G) \backslash S]$ that remain after removing $S$ in increasing order after $S$. Within each connected component, place vertices recursively using \textsc{Separator-LA}.
\end{enumerate}

\begin{lemma}\label{lem:mla-separator}
\textsc{Separator-LA}$(G)$ produces a linear arrangement of cost $O(n \Delta s(G) \log n)$. If $s_n(G)\in \Theta(n^{\epsilon})$ for some constant $\epsilon>0$, the linear arrangement has cost $O(n\Delta s_n(G))$. 
\end{lemma}
\begin{proof}
The cost at a particular level of recursion is at most $n \Delta |S|  \leq n \Delta s_n(G)$. Let $n_0, \dotsc, n_i$ be the sizes of the connected components of the graph $G[V(G) \backslash S]$ after removing $S$. The cost $\lambda(n)$ of the linear arrangement on $n$ vertices is at most:
\[
\lambda(n) \leq O(n \Delta s_n(G)) + \sum_i \lambda(n_i) \enspace ,
\]
which solves to $O(\Delta n s_n(G) \log n)$ because the depth of the recursion is $O(\log n)$. Note that if $s_n(G)\in \Theta(n^{\epsilon})$,  $\sum_i s_{n_i}(G) < s_n(G)/2$.
\end{proof}

See \Cref{tab:mla-bounds} for a summary of the bounds obtained by \textsc{Separator-LA} on various families of graphs. 


\begin{table}[]
    \caption{Bounds on the cost of a linear arrangement.}
    \centering
    \vspace{-1em}
    \def\arraystretch{1}
    \setlength\tabcolsep{0.5em}
    \begin{tabular}{ll}
    \toprule
     Graph family & Linear arrangement cost\\ \midrule
     $K_r$-minor free   (\S \ref{sec:separator-la}, ~\cite{DBLP:conf/focs/KawarabayashiR10}) & $O(n \Delta \sqrt{n} r)$  \\
     Planar  (\S \ref{sec:separator-la}, ~\cite{doi:10.1137/0136016})  & $O(n \Delta \sqrt{n})$ \\
     Treewidth $\tau$ (\S \ref{sec:separator-la} ~\cite{DBLP:journals/ejc/BottcherPTW10, ROBERTSON1986115})  & $O(n\tau \log n)$ \\
     Series-Parallel (\S \ref{sec:separator-la})  & $O(n\log n)$ \\
     Trees (\S \ref{sec:tree-arrow-decomposition}) & $n\Delta$ \\
     \bottomrule
    \end{tabular}

    \label{tab:mla-bounds}
\end{table}

\subsection{Linear Arrangements using Random MSTs}\label{sec:mla-random-forests}

For datasets with hundreds of millions of nodes, it is crucial to have an algorithm that uses \emph{near-linear} time. Computing separators in $K_r$-minor-free graphs takes $\Omega(m\sqrt{n})$ time using state-of-the-art algorithms~\cite{DBLP:conf/soda/PlotkinRS94}, which becomes prohibitive for graphs with hundreds of millions of vertices.
We propose a linear arrangement scheme using a \emph{random spanning forest} of the input graph: 
\begin{enumerate}[leftmargin=0.5cm]
\item Construct a weighted graph $G'$ by drawing edge weights independently from the standard uniform distribution.
\item Compute a minimum spanning forest $F$ of $G'$.
\item Compute a linear arrangement of each tree in the forest $F$ in decreasing order of size and concatenate them.
\end{enumerate}
In our experiments in Section \ref{sec:experiments}, we evaluate the linear arrangement using random forests and demonstrate its efficacy on real-world datasets. As trees have separation number $2$, we directly get a linear arrangement of the spanning trees of cost $O(n\Delta \log n)$ using \textsc{Separator-LA}. However, improving the quality of the linear arrangement of trees is possible. 

\subsection{Linear Arrangement of Trees}\label{sec:tree-arrow-decomposition}

In this section, we show an improvement in the cost of a linear arrangement over \textsc{Separator-LA} by a factor of $\Theta(\log n)$ for trees. We can get a tighter bound on the arrow width of an arrow decomposition of a tree with the following layout, which we use in our experiments: Place the root at the first position. Then, sort the children subtrees by size and arrange these subtrees one after the other in this order. Arrange each subtree recursively. This arrangement $\pi$ is called \emph{smallest-first order}. Instead of arguing about the cost of the linear arrangement and then using that most edges are close to the average, we directly argue about how many edges are within a $x\Delta$ wide band around the diagonal. For every edge $(u, v)$ in the $x\Delta$-wide band around the diagonal, we have that $|\pi(u)-\pi(v)| \leq x\Delta$.

\begin{lemma}\label{lem:smallest-first-order}
In smallest-first order $\pi$ of a tree $T$, at least
$$\min\left(n-1, \lceil \frac{(x-1)\cdot (n - 1)}{x}\rceil + 1\right)$$
edges are within an $x\Delta$-wide band around the diagonal.
\end{lemma}
\begin{proof}
Observe that the vertices of every subtree are listed consectively in the linear order $\pi$. Hence, we can use strong induction on the number of edges in the tree $T$. We root the trees at an arbitrary vertex.  
For a vertex $v$, let $E(v)$ be the set of edges in the subtree rooted at $v$. Let $E_{x}(v)$ be the set of edges in $E(v)$ within the $x\Delta$ band around the diagonal. 
Let $P(v)$ be the predicate
\begin{align*}
P(v) \equiv &\text{ If $|E(v)| \geq x$, then at least } \left \lceil \frac{x - 1}{x}|E(v)| \right \rceil + 1 \text{ edges} \\
&\text{in $E(v)$ are within a }
 x\Delta \text{ band around the diagonal}	\enspace .
\end{align*}
Note thay if the tree has less than $x$ edges, then all its edges are within an $x\Delta$ band around the diagonal. 
We prove inductively that $P(v)$ holds for all trees. 
As a base case, consider an arbitrary tree rooted at $v$ with $x \leq E(v) \leq x\Delta$ edges. For such a tree, every edge is within a $x\Delta$ band around the diagonal and $|E(v)| \geq \lceil \frac{x - 1}{x} \cdot |E(v)|\rceil + 1$. 
For the inductive step, consider consider a tree rooted at $v$ where $|E(v)| > x \Delta$. By induction, for each $v'\neq v$ in the subtree rooted at $v$ we may assume that $P(v')$ holds. We proceed by case distinction on the degree of $v$.\\

\noindent
\textbf{Case} $\mathbf{\deg(v) = 1}$
Let $u$ be the child of $v$. By definition of \emph{smallest-first} order, the distance between $u$ and $v$ in the linear arrangement $\pi$ is $1$ and therefore the $\{v, u\} \in E_{x}(v)$. Notice that $E_{x}(v) = \{v, u\} \cup E_{x}(u)$. Because $|E(u)|=|E(v)|-1 \geq x $, we conclude by $P(u)$ that
\begin{align*}
	E_{x}(v) &= 1 + |E_{x}(u)| \\
	&\geq 1 + \left\lceil\frac{x - 1}{x} \cdot |E(u)|\right\rceil + 1 \\
	&\geq \left\lceil \frac{x - 1}{x} \cdot |E(v)|\right\rceil + 1 \enspace .
\end{align*}
\textbf{Case} $\mathbf{\deg(v) \geq 2}$.
Let $C(v) = u_1, ..., u_{\deg(v)}$ be the list of children of $v$, sorted in increasing order by the size of their subtree. Let $i$ be the largest index such that $|\pi(v) - \pi(c_i)| \leq x\Delta$, i.e., $\{v, u_i\}$ is in the $x\Delta$ band. Notice that $\forall j \leq i \enspace |\pi(v) - \pi(u_j)| \leq x\Delta$, by the definition of \emph{smallest-first} order. Because $v$'s subtree is of size greater than $x\Delta$, we have $i < \deg(v)$. It now follows that:
\begin{align*}&|E_{x}(v)|  \\
        = \ &i + \sum_{w = 1}^{\deg(v)}|E_{x}(u_w)|\\
        = \ &\left(\sum_{w = 1}^{i}|E_{x}(u_w)| + 1 \right) + \sum_{w = i + 1}^{\deg(v)}|E_{x}(u_w)|\\
        \geq \ & \left(\sum_{w = 1}^{i}\left\lceil\frac{x-1}{x}|E(u_w)|\right\rceil + 1\right) + 
        \sum_{w = i + 1}^{\deg(v)}\left\lceil\frac{x-1}{x}|E(u_w)|\right\rceil + 1\\
        \geq \ & \sum_{w = 1}^{\deg(v)}\left\lceil\frac{x-1}{x}|E(u_w)|\right\rceil + 1\\
        \geq \ & \deg(v) + \sum_{w = 1}^{\deg(v)}\left\lceil\frac{x-1}{x}|E(v)|\right\rceil\\
        \geq \ & \deg(v) + \frac{x-1}{x} |E(u_v)| \\
        \geq \ & \left\lceil\frac{x-1}{x}|E(v)|\right\rceil + 1 \enspace .
\end{align*}
We now explain the first inequality. First, we look at vertices in $\sum_{w = 1}^{i}\left(|E_{x}(u_w)| + 1 \right)$. If $|E(u_w)| \geq x$, we have by induction hypothesis that $|E_{x}(u_w)| + 1 \geq \left\lceil\frac{x-1}{x}|E(u_w)|\right\rceil + 1$. If $u_w$ has less than $x$ edges in its subtree, we cannot use the induction hypothesis, but we still have $|E_{x}(u_w)| + 1 = |E(u_w)| + 1 \geq \left\lceil\frac{x-1}{x}|E(u_w)|\right\rceil + 1$. Next, observe that $\sum_{w = 1}^i |E(u_w)| \geq x\Delta$ because otherwise $\{v, w_{i + 1}\}$ would be in the $x\Delta$ band. This means that at least one child $u_w$ with $w \leq i$ has to have at least $x$ edges in its subtree. Because the subtrees are sorted by size, it follows that all children $u_w$ with $w > i$ satisfy $|E(u_w)| \geq x$ and we can use the induction hypothesis. 
.
\end{proof}
We immediately get a more efficient $x$-compacting arrow matrix decomposition for trees using \textsc{LA-Decompose}:

\begin{corollary}\label{lem:arrow-tree-decomposition}
	A tree has an $x$-compacting $x\Delta$-arrow decomposition that can be computed in $O(n)$ work.
\end{corollary}
\begin{proof}
Follows from \Cref{lem:la-decompose} and \Cref{lem:smallest-first-order}.
\end{proof}
Note how this result contrasts with the bounds on the bandwidth of a tree graph: The bandwidth of a balanced binary tree is $\Omega(n/\log n)$, whereas, we can construct a decomposition into $O(\log n)$ matrices of bandwidth $O(1)$.

\subsection{Lower Bounds}\label{sec:lower-bounds}

The linear dependence on the maximum degree $\Delta$ is necessary for any linear arrangement of the graph families listed in \Cref{tab:mla-bounds}. We prove the lower bound for trees first, which then implies the other lower bounds:
\begin{lemma}
For every $\Delta>3$, there are trees with a minimum linear arrangement of cost $\Omega(n\Delta)$.	
\end{lemma}
\begin{proof}
First, consider a star graph of $\Delta-1$ nodes. Any linear arrangement costs at least $\Omega(\Delta^2)$, as, at least a quarter of the nodes are at least $\Delta/4$ away from the central node (no matter where it is placed). Moreover, observe that inserting additional nodes into the graph can only increase the cost incurred by the edges in the star.

Now, consider a complete $(\Delta-1)$-ary tree. The parents of the leaf nodes together with their descendants constitute $\Omega(n/\Delta)$ disjoint star graphs with degree $\Delta-1$. Their layout costs $\Omega(\Delta^2)$ each, which implies the result.
\end{proof}
The  $\Omega(n\Delta)$ lower bound on the cost of a linear arrangement applies to all families in \Cref{tab:mla-bounds} and is tight for trees, as shown in \Cref{sec:tree-arrow-decomposition}. The linear dependence on the maximum degree is undesirable, as many sparse real-world graphs exhibit a large maximum degree~\cite{RevModPhys.74.47}. 
Next, we show the arrow decomposition overcomes this dependence on the maximum degree by pruning the highest-degree vertices.

\subsection{Pruning in Power Law Graphs}
\label{sec:pruning}

Many real-world graphs, such as the web graph, social networks, and protein interaction networks, exhibit a power law degree distribution~\cite{RevModPhys.74.47}. This means that while the average degree is small, the maximum degree can be a significant fraction of the number of vertices. On these graphs, the first step of \textsc{LA-Decompose} (pruning the highest degree vertices) provides a \emph{polynomial} improvement in the arrow width. We proceed to bound the improvement analytically as a function of the power law. 

There are various probability distributions that generate a power law~\cite{balakrishnan2004primer,KOZUBOWSKI2015135,zipf1949human,22AddoTruncatedZipf}. To model the vertex degrees, it is appropriate to choose a discrete distribution that is bounded to the interval of the number of vertices. Hence, we model the degree distribution of a vertex as a \emph{discrete truncated Zipf distribution}~\cite{22AddoTruncatedZipf}, truncated between $1$ and $n$ with shape parameter $\alpha$. Note that for simplicity of notation, we are considering $n+1$ vertices here giving degrees between $1$ and $n$. We exclude the possibility of singleton vertices, as they do not contribute to the arrow width. 

The probability mass function $p(x)$ of a discrete truncated Zipf distribution is given by
\begin{align}
    p(x) = \frac{x^{-\alpha}}{\sum_{j=1}^{n} j^{-\alpha}} \enspace .
\end{align}
The term in the denominator is the generalized harmonic number $H_{n, \alpha}$. Note that as $n$ goes to infinity, the generalized harmonic numbers approach the Riemann zeta function $\zeta(\alpha)=\sum_{j=1}^{\infty} j^{-\alpha}$.  For an integer $x\geq 0$, the survival probability $S(x)$ is given by
\begin{align}
S(x) = \frac{H_{n, \alpha}-H_{x, \alpha}}{H_{n, \alpha}} \enspace .
\end{align}
The expected number of vertices with degrees larger than some given $x$ is at most $nS(x)$. This tells us how many vertices we need to prune (in expectation) to be left with a graph with maximum degree $x$. To derive a bound on this expectation $nS(x)$, we derive a closed-form approximation to the survival function:

\begin{theorem}
\label{thm:zipf-survival-function}
For all $x$ larger than some constant, the survival function $S(x)$ of the truncated Zipf distribution with shape $\alpha>1$ truncated between $1$ and $n$ is bounded by $S(x)\leq \frac{x^{1-\alpha}}{(\alpha-1)\zeta(\alpha)}$.
\end{theorem}
\begin{proof}
We lower bound the cumulative distribution function $F(x) = 1-S(x) = \frac{H_{x,\alpha}}{H_{n,\alpha}}$, which gives us an upper bound on $S(x)$. The main technical challenge is to obtain a suitable closed-form approximation to the generalized harmonic numbers. We employ the Euler-Maclaurian summation formula~\cite{graham94ConcreteMath} to bound $\zeta(\alpha)-H_{n, \alpha}$, which implies that for any constant $\alpha>1$ 
\begin{align*}
H_{n, \alpha} = \zeta (\alpha) +  \frac{n^{1-\alpha}}{1-\alpha}  + \frac{n^{1-\alpha}}{2n}  - \frac{\alpha n^{1-\alpha}}{12n^2}  + O\left(\frac{\alpha n^{1-\alpha}}{n^3} \right) \enspace .
\end{align*}
For large enough $x$ and $x+1 \geq \alpha>1$, the first two terms dominate:
\begin{align*}
    H_{x,\alpha} \geq \zeta (\alpha) +  \frac{1}{1-\alpha} x^{1-\alpha} \enspace , \\
    H_{x,\alpha} \leq \zeta (\alpha) +  \frac{1}{2(1-\alpha)} x^{1-\alpha} \enspace . 
\end{align*}

\noindent
Using these inequalities we can proceed:
\begin{align*}
    F(x) \geq & \ \frac{ \zeta (\alpha) +  \frac{1}{1-\alpha} x^{1-\alpha} }{ \zeta (\alpha) +  \frac{1}{2(1-\alpha)} n ^{1-\alpha} } \enspace \\
    = & \ \frac{2(\alpha-1)n^{\alpha-1} \zeta (\alpha) - 2\frac{n^{\alpha-1}}{x^{\alpha-1}}}{2(\alpha-1)n^{\alpha-1} \zeta (\alpha) - 1} \enspace \\
    \geq & \ 1 - \frac{x^{1-\alpha}}{(\alpha-1) \zeta (\alpha)}  \enspace , 
\end{align*}
which implies the result.
\end{proof}
Note that this implies that the survival function itself takes on the shape of a power law. The larger $\alpha$, the quicker the survival function diminishes. We are now ready to bound the number of high-degree vertices in a power law graph.
\begin{lemma}\label{lem:zip-degree-bound}
Consider a graph $G$ whose degree distribution follows a truncated Zipf distribution with shape parameter $\alpha>1$. For any $b\geq \Omega(1)$ and $\Delta_0\geq \Omega(1)$, the probability that $G$ has more than $b$ vertices of degree larger or equal to $\Delta_0$ is at most $\frac{n \Delta_0^{1-\alpha}}{b(\alpha-1) \zeta(\alpha)}$.
\end{lemma}
\begin{proof}
The expected number of vertices with degrees larger than $\Delta_0$ is at most $nS(\Delta_0)$. The result follows from \Cref{thm:zipf-survival-function} and Markvov's inequality.
\end{proof}

Let us see what this implies for the question of pruning high-degree vertices. 
If we set $\Delta_0=n^{\delta}$ for some constant $\delta>0$, we get that after pruning the $b\in \omega(n^{(1-\alpha)\delta + 1})$ vertices of largest degree, the maximum degree of the remaining subgraph is at most $\Delta_0$ with probability $1-o(1)$. As our bounds on the cost of a linear arrangement depend linearly on the maximum degree of the subgraph that remains after pruning the high-degree graphs, we would like to balance the number of pruned vertices with the remaining maximum degree. The parameter $\delta$ that achieves this balance is $\delta=\frac{1}{\alpha}$. 

We conclude with the implication of this result for the arrow decomposition of trees and note that we can derive similar statements for the other considered graph families:
\begin{corollary}
Consider a tree whose degree distribution follows a truncated Zipf distribution with shape $\alpha>1$. \textsc{LA-Decompose} with parameter $b=\omega( n^{\frac{1}{\alpha}})$ produces an $x$-compacting $xb$-arrow matrix decomposition with probability $1-o(1)$.
\end{corollary}
\begin{proof}
Follows from \Cref{lem:zip-degree-bound} and \Cref{lem:arrow-tree-decomposition}.
\end{proof}

Observe that this bound is now independent of the maximum degree in the original graph, which would have been $\Omega(n)$ in expectation. Hence, pruning the high-degree vertices provides a polynomial improvement in the arrow width of power law graphs.

\section{Data Movement Analysis}
\label{sec:spmm}

Sparse matrix multiplication is a typical memory-bound operation when the dense matrix is tall and skinny, as the number of arithmetic operations is of a similar order of magnitude to the number of memory accesses. Hence, minimizing data movement is paramount to achieving the best performance and scalability. Note that for sparse datasets, the size of the feature matrix $\mat{X}$ dominates the storage, i.e., $m \ll nk$. Our algorithm falls into the class of $\mat{A}$-stationary algorithms, where the sparse matrix remains local and only the dense feature matrix and the result of the SpMM are communicated. 
%
We show that at the cost of a slightly increased latency, a $c$-compacting arrow decomposition enables a $\Theta(\sqrt{p})$ reduction in bandwidth requirements compared to a direct 1.5D decomposition and a $\Theta({\sqrt{p}})$ storage improvement in the setting where the feature matrix dominates the storage.

\subsection{Data Movement}

\paragraph{Multiplying with an Arrow Matrix}

We now analyze the communication cost incurred by our approach from Section 3.1 in the $\alpha-\beta$ model of computation. We focus first on the multiplication of an arrow matrix $\mat{B}$ with a tall-skinny matrix $\mat{X}$.

\begin{lemma}\label{lem:distributed-arrow-spmm}
Consider a matrix $\mat{B}\in R^{n\times n} $ with arrow-width $b$ and a matrix $\mat{X} \in R^{n\times k}$. If $p=\ceil{n/b}$, computing $Y=\mat{BX}$ has a communication cost of $O(\alpha \log p + \beta\;bk \log p)$.
\end{lemma}

\begin{proof}
Recall that we have $\ceil{\frac{n}{b}}$ row tiles $\mat{B_{0,j}}$, $\ceil{\frac{n}{b}}$ column tiles $\mat{B_{i,0}}$ and $\ceil{\frac{n}{b}}$ diagonal tiles $\mat{B_{i,i}}$. Let $\mat{X}$ also be split row-wise into $\ceil{\frac{n}{b}}$ blocks of size $b \times k$. We distribute the calculation of $\mat{BX}$ as follows: For each row tile $\mat{B_{0,j}}$, there is a processor responsible for calculating $\mat{B_{0,j}}\mat{X_{j}}$. The intermediate results are reduced and summed at one node. For each column tile $\mat{B_{i,0}}$ with $i>0$, we have a processor responsible for calculating $\mat{B_{i,0}}\mat{X_{0}} + \mat{B_{i,i}}\mat{X_{i}}$. Due to the arrow shape, we only have two non-zero tiles per row when $i>1$. Hence, we can do the calculation of the entire row on one processor. 

Overall, we have $\ceil{\frac{2n}{b}}-1$ computation tasks which we assign to our $p$ processors. 
We assume that the tiles of $\mat{B}$ are already correctly distributed as they remain fixed throughout the iterations. Note that half of the computation tasks will require a copy of $\mat{X_0}$, i.e. we will need one broadcast of $\mat{X_0}$ to half of the processors which incurs a communication cost of $O(\alpha \log p + \beta\;bk \log p)$. The reduce-operation for the row tiles incurs the same cost (the reduce-operation also involves $\frac{p}{2}$ processors since the row tiles make up half of the computation tasks). Lastly, we will need to send the diagonal blocks of $\mat{X}$ to the right processors using pairwise communication. Since each processor only needs a single block from $\mat{X}$, this only incurs a cost of $O(\alpha + \beta\;bk)$.
\end{proof}

\paragraph{Multiplying with an Arrow Decomposition}

Next, we analyze the communication cost of combining the intermediate products $\mat{B}_i ( \mat{P}_{\pi_i} ^{\top}  \mat{X} )$ of our matrix decomposition to finally arrive at $\mat{Y}=\mat{A}\mat{X}$. To improve the efficiency of multiplying repeatedly with the same matrix $\mat{A}$, we leave the rows of $\mat{Y}$ permuted in the order of the first matrix in the decomposition. In the end, we might need to permute back to the original order of rows depending on the application. If the result is required in the original order of rows, the communication cost of this permutation is fully amortized after at most $\log^2 p$ iterations of multiplying with $\mat{A}$. 
The key insight is that when the number of nonzeros decreases quickly, we can implement the permutations for the aggregation more efficiently than doing a naive all-to-all. 
\begin{theorem}\label{lem:distributed-matrix-iteration}
For a matrix $\mat{A}\in R^{n\times n}$ with an $x$-compacting $b$-arrow decomposition and a feature matrix $\mat{X} \in R^{n\times k}$, if $x\geq \Omega(\log ^2 p)$ and $p=\Theta(\frac{n}{b})$,  computing $\mat{P}_{\pi_0}^{\top}\mat{Y}=\mat{P}_{\pi_0}^{\top}\mat{A}\mat{X}$ has a communication cost of $O(\alpha \log^2 p + \beta \frac{nk}{p})$. 
\end{theorem}
\begin{proof}
    Assuming we calculated $\mat{Y}_i=\mat{B}_i ( \mat{P}_{\pi_i} ^{\top} \mat{X} )$ for each $\mat{B_i}$ matrix in our decomposition using \Cref{lem:distributed-arrow-spmm}, it remains to aggregate the results. 
We need to sum the $\mat{Y_i}$ matrices, however, each of them has its rows permuted differently. Let $\mat{Y}_i$ and $\mat{Y}_{i+1}$ be any two of these matrices that we want to sum and let $P_i$ be the set of processors among which the blocks of $\mat{Y}_i$ are distributed.

If we were to send all the rows of $\mat{Y_{i+1}}$ to their corresponding processor in $P_i$ naively, in the worst case we would have to perform one scatter operation for each processor in $P_{i+1}$ to all the processors in $P_i$. 
To alleviate this, we will first sort the rows of  $\mat{Y_{i+1}}$ by their destination processor in $P_i$. Then, for any $p^{i+1}_j \in P_{i+1}$, it holds that if it stores rows of $\mat{Y_{i+1}}$ that need to be sent to processors $p^{i}_t... p^{i}_{t+k}$ then for any processor $p^{i+1}_l$ with $l>j$, it holds that it only stores rows of $\mat{Y_{i+1}}$ that need to be sent to $p^{i}_{m}... p^{i}_{m+k'}$ where $m \geq t+k$. 

Note that we can determine the destination processor for each row in advance since the permutation matrices are fixed in the pre-processing. We can also pre-determine the destination processor range of each processor in $P_{i+1}$ this way. After sorting, we can schedule the scatter operations much more efficiently in parallel: Observe that each processor logically receives a message from at most two scatter operations and that these come from neighboring processors in $P_{i+1}$. Hence, we can perform the scatter operations in two phases involving first all evenly-indexed processors in $P_{i+1}$ and then oddly-indexed processors in $P_{i+1}$.
This scattering then incurs a communication cost of $O(\alpha \log p + \beta\; \frac{nk}{c} \log p)$.

As for the sorting itself, we can employ a sorting network and place each processor of $P_{i+1}$ on a wire. When two processors $p^{i+1}_k$ and $p^{i+1}_l$ with $k < l$ are connected, $p^{i+1}_k$ will send $p^{i+1}_l$ all its rows that are above its range and $p^{i+1}_l$ will send $p^{i+1}_k$ all its rows that are below. For a bitonic sorting network with depth $O(\log^2 p)$, this leads to a communication cost of  $O(\alpha \log^2 p + \beta\; \frac{nk}{c} \log^2 p)$. Sorting networks with a smaller depth exist but are unwieldy in practice.

Overall, we can thus perform the sorting and the subsequent scattering with a communication cost of $O(\alpha \log^2 p + \beta\; \frac{nk}{c} \log^2 p)$. Assuming we have a $c$-compacting decomposition, recall that it holds that the total number of non-zeros of $Y_{i+1}$ is at most $\frac{1}{c}$ times that of $Y_{i}$. Thus, if we perform the summation of all $l$ parts in decreasing order, the total communication can be upper-bounded by the summation of $Y_0$ and $Y_1$, i.e., $O(\alpha \log^2 p + \beta\; \frac{nk}{c} \log^2 p)$. For $c\geq \Omega(\log ^2 p)$, we arrive at a communication cost of $O(\alpha \log^2 p + \beta \frac{nk}{p})$ for one matrix iteration, as desired.

To prepare the next iteration, the accumulated result is distributed in the reverse communication pattern (replacing scatters with gather). This concludes the calculation of $\mat{P}_{\pi_0}^{\top}\mat{A}\mat{X}$.
\end{proof}
This result shows that given an adequate arrow decomposition, we obtain a communication cost that improves on a fully replicated 1.5D decomposition by a factor $\sqrt{p}$ in terms of bandwidth at the cost of a $\log p$ factor in terms of latency. Compared to a 1D decomposition, the latency cost is a factor $\Omega(\frac{p}{\log p})$ smaller. Compared to an $\mat{A}$-stationary 2D decomposition, the bandwidth cost is a factor $\Theta(\sqrt{p})$ smaller and the latency cost is a factor $\Theta\left(\frac{\sqrt{p}}{\log n}\right)$ smaller.

\subsection{Storage Requirements}

We propose to store the matrix $\mat{A}$ in a sparse format, such as compressed sparse row (CSR), and store the dense matrices in row-major.
Because an arrow matrix has fewer blocks compared to a matrix that has been 1.5D decomposed directly, we can afford to have smaller blocks with the same number of processors. This removes the main downside of a 1.5D decomposition with full replication, namely its high storage requirements.
\begin{lemma}\label{lem:arrow-storage}
Consider a matrix $\mat{A}\in R^{n\times n} $ with an $x$-compacting $b$-arrow decomposition and $\mat{X} \in R^{n\times k}$. If $x \geq 1+\epsilon$ for a constant $\epsilon$ and the blocks of $\mat{A}$ are stored in CSR, the total storage cost is $m + O(nk)$.
\end{lemma}
\begin{proof}
The number of nonzero rows in the matrices of the decomposition decreases geometrically because the decomposition is $x$-compacting. Hence, the cost to store $\mat{X}$ is $O(nk)$ and the cost to store row and index pointers for $\mat{A}$ is $O(n)$. Since each edge occurs in exactly one matrix of the decomposition, the value arrays take up $m$ space.
\end{proof}

The storage matches that of a 2D decomposition. Compared to a 1.5D decomposition with replication factor $c$, the storage used by the dense matrices is a factor $\Theta(c)$ smaller. For full replication, this constitutes a factor $\Theta(\sqrt{p})$ saving. In conclusion, the arrow matrix decomposition shares the favorable space requirements of 1D and 2D decompositions while improving on the bandwidth cost and memory requirements of a fully replicated 1.5D decomposition.


\section{Evaluation}\label{sec:experiments}

We benchmark the scalability of sparse matrix multiplication with our decomposition compared to a 1.5D decomposition and a 1D hypergraph partitioning decomposition. 

\subsection{Setup}

\minisec{System}
 We ran our SpMM experiments on the Piz Daint supercomputer on the Cray XC50 nodes with Aries routing and a dragonfly network topology. Each node has a 12-core HT-enabled Intel Xeon E5-2690 v3 CPU with 64GB of RAM and one NVIDIA Tesla P100 GPU with 16GB of memory.
For decomposing the graphs, we ran our algorithm on single nodes with 4 Intel(R) E7-4830 v4 CPUs and 2TB of memory, and 2 AMD EPYC 7742 CPUs with 512GB of memory.

\minisec{Datasets}
We evaluate our approach on sparse graph datasets with 18-226 million rows and up to a $1.9$ billion nonzeros from the SuiteSparse Matrix Collection~\cite{10.1145/2049662.2049663}. 
See \Cref{tab:dataset-properties} for a summary of the dataset's characteristics. We considered all matrices with at least 50M rows and fewer than 100 nonzeros per row on average. The denser graphs cause out-of-memory issues and timeouts with both the baselines and our approach. 
For the feature matrices, we use $128$, $64$, and $32$ columns. 

\begin{table}[]
\caption{Summary of the datasets' density properties.}
\vspace{-0.8em}
    \centering
    \def\arraystretch{0.97}
    \setlength\tabcolsep{0.4em}
    \begin{tabular}{llll}
    \toprule
    Dataset & Vertices $n$ & $\frac{\text{nnz}(A)}{n}$ & Max. Degree $\Delta$ \\ \midrule
    MAWI 226M & 226,196,185 & 2.12 & 210,795,477  \\
    MAWI 69M & 68,863,315 & 2.08 & 63,040,326  \\
    GenBank 214M & 214,005,017 & 2.17 & 8  \\
    GenBank 68M   & 67,716,231 & 2.05 & 35  \\
    WebBase 118M & 118,142,155 & 8.63 & 816,127  \\
    OSM Europe & 50,912,018 & 2.12 & 13 \\
    GAP-twitter 62M & 61,578,415 & 23.85 & 770,155 \\
    sk-2005 51M &  50,636,154  & 38.50 & 8,563,808 \\
    \bottomrule
    \end{tabular}
    \vspace{-1em}
    \label{tab:dataset-properties}
\end{table}

\minisec{Implementation} We implement all SpMM algorithms as a Python module, using \textit{numpy} v1.24.2~\cite{harris2020array}, \textit{scipy} v1.10.1~\cite{2020SciPy-NMeth}, \textit{cupy} v11.6~\cite{cupy_learningsys2017}, and \textit{mpi4py} v3.1.4~\cite{DALCIN20051108}. 
The single-block intra-node GPU SpMM operations are implemented with the \emph{CSR}-times-dense matrix multiplication (CSRMM) kernels found in the 
NVIDIA cuSPARSE (v11.0). Our experiments use Cray-MPICH v7.7.18. 
The code is available at: \href{https://github.com/spcl/arrow-matrix}{https://github.com/spcl/arrow-matrix}
%

Our approach utilizes the linear arrangement framework (\Cref{sec:la-decompose}) with pruning (\Cref{sec:pruning}). We construct the linear arrangements using the random spanning MSTs algorithm (\Cref{sec:mla-random-forests}). For each tree, we employ the algorithm from \Cref{sec:tree-arrow-decomposition}. The decomposition uses \textsc{Julia} v.1.9.3. Our implementation may leave a few ranks unused when the block size does not evenly divide the matrix size.   

\minisec{1.5D Baseline}
We compare our approach based on the 1.5D decomposition of an arrow matrix decomposition against the direct 1D and 1.5D decomposition schemes. We use the same libraries and kernels to ensure a comparison that focuses on the merits of the decomposition schemes. Similarly as Tripathy at et al.~\cite{DBLP:conf/sc/TripathyYB20} we include a parameter $c$ that interpolates between the $1D$ decomposition ($c=1$) and the single-round 1.5D decomposition ($c=\sqrt{p}$). As shown in \Cref{fig:15d_baseline_c}, a larger $c$ leads to lower runtimes, as expected from the theory. 
Hence, we use $c=\left\lfloor\sqrt{p}\right\rfloor$ in our experiments, resulting in one or two computational rounds. 
For larger feature sizes, the whole data sometimes do not fit in GPU memory. In this case, we further tile the single-round computation into smaller blocks. 
This approach is significantly faster than lowering $c$ since host-to-device memory transfers are faster than the network.

\begin{figure}
    \centering
  \begin{minipage}[t]{0.15\textwidth} 
    \centering
    \includegraphics[width=\textwidth, trim=0.5cm 0cm 0cm 0cm]{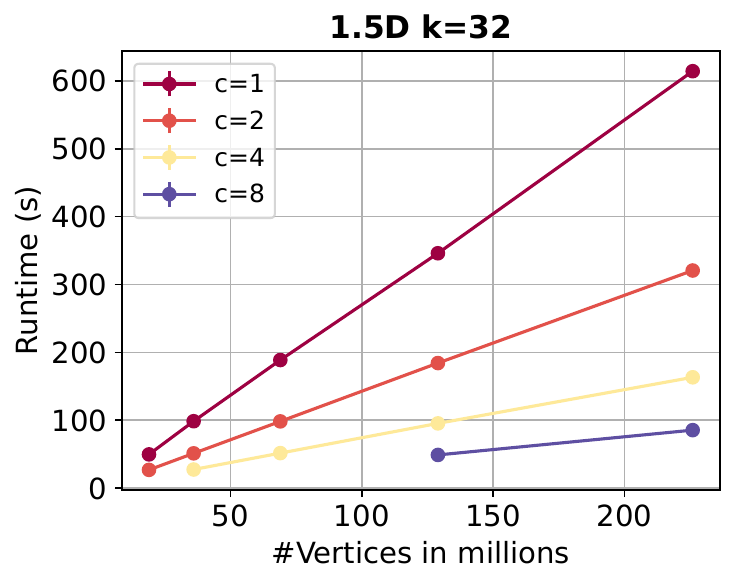}
  \end{minipage}
  \begin{minipage}[t]{0.15\textwidth} 
    \centering
    \includegraphics[width=\textwidth, trim=0.5cm 0cm 0cm 0cm]{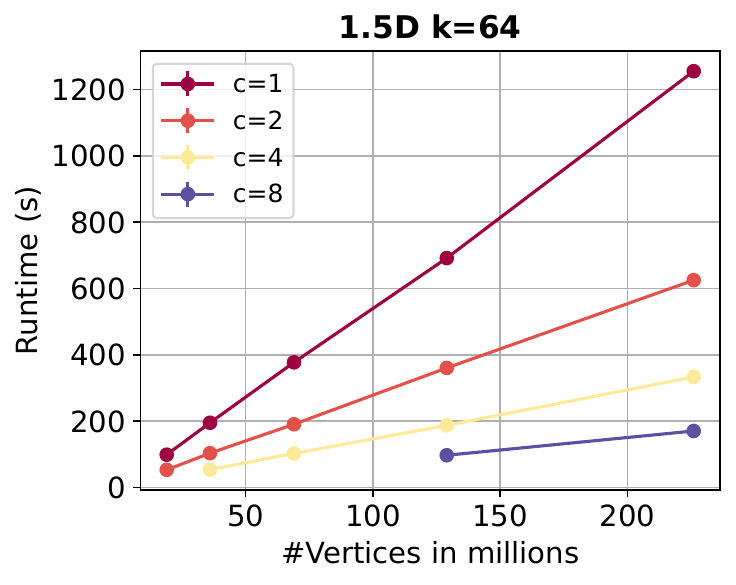}
  \end{minipage}
  \begin{minipage}[t]{0.15\textwidth} 
    \centering
    \includegraphics[width=\textwidth, trim=0cm 0cm 0.5cm 0cm]{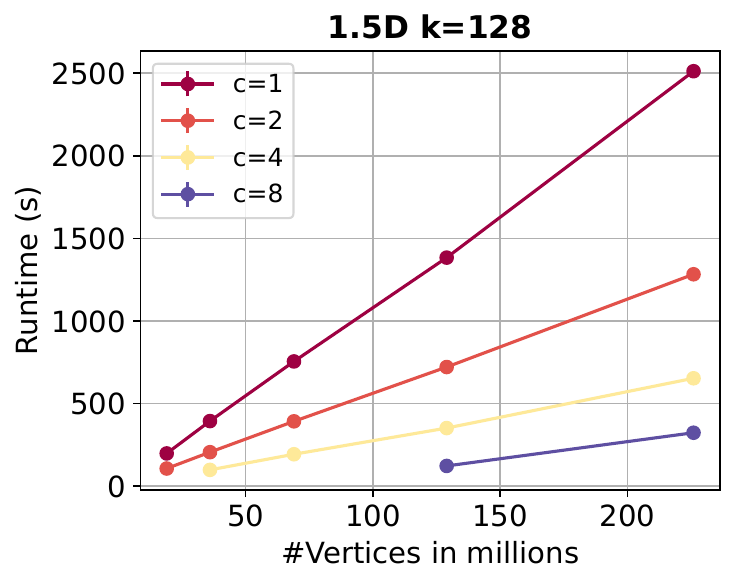}
    \label{fig:figure2} 
  \end{minipage}
  \vspace{-2em}
    \caption{Weak scaling of the 1D / 1.5D baseline for varying replication factors $c$ on the MAWI datasets.}
  \vspace{-1em}
    \label{fig:15d_baseline_c}
\end{figure}
\minisec{Hypergraph Partitioning Baseline} Additionally, we compare our implementation against a 1D hypergraph partitioning scheme (HP-1D). We adapt the PETSc-style variant from previous work on SpMV~\cite{DBLP:conf/ppam/KayaUC13} to the SpMM setting. The matrix is permuted according to a hypergraph partitioning. The hypergraph contains a vertex for each row $i$ and a net for each column $j$ that connects all vertices $i$ for which the matrix has a nonzero entry in row $i$ and column $j$. To partition the hypergraphs, we use HYPE~\cite{DBLP:conf/bigdataconf/MayerMBER18}. After permuting, the matrices are split row-wise in 1D. The computation consists of two parts; a local SpMM that requires no communication and a non-local SpMM for which features of other processors are needed. The local SpMM can overlap with the message exchange, hiding computational costs. We implement this overlap using MPI nonblocking send and receive. 

\minisec{Measurements}
We run each SpMM for at least $7$ iterations and drop the first iteration, as it includes GPU and library initialization costs. We report the mean over the iterations of the maximum runtime of any participating rank. Additionally, we show the minimum and maximum over the iterations when they deviate more than $5\%$ from the mean.

\subsection{Decomposition Results}

We executed our random forests algorithm on each dataset, varying the arrow width as $b\in {0.5 \times 10^6, 1 \times 10^6, \dots, 5 \times 10^6 }$. This approach yielded, at most, $4$ matrices in the decomposition for all datasets. Additionally, the second matrix in the decomposition contained between a few hundred and $25$ million nonzero rows, constituting less than $0.1\%-13\%$ of the rows. This aligns with our assumptions in \Cref{lem:distributed-matrix-iteration}, affirming our algorithm's generation of highly compact decompositions for these sparse datasets.

Refer to Figure \ref{fig:decomposition-illu} for a depiction of the first matrix in the decomposition (with analogous outcomes for all MAWI and GenBank matrices). Notably, the datasets exhibit distinct characteristics: In the MAWI data, most nonzeros cluster in the 'pruned' part near the matrix's top and left corners. Conversely, for the GenBank and OSM Europe data, the majority of nonzeros appear in the diagonal band. The Webbase and GAP-twitter datasets showcase a notable number of nonzeros in the pruned segment, yet they are less skewed than the MAWI datasets. These patterns correspond with the datasets' maximum degrees. Specifically, the MAWI datasets feature a nearly equal maximum degree and vertex count due to the prevalence of large star subgraphs. In contrast, the Webbase and GAP-twitter datasets have a maximum degree of $8-13\%$ of the vertex count. The GenBank k-mer and OSM Europe datasets, on the other hand, display a maximum degree at most $20$ times the average. Our approach handles these diverse matrix behaviors robustly. 

\minisec{Comparison with 1.5D}
%
We contrast the count of nonzero blocks in our arrow decomposition matrices with those in a 1.5D decomposition employing equally sized blocks. As matrix rows increase, our method uses notably fewer nonzero blocks. For the two largest datasets, our decomposition results in $15$-$20$ times fewer nonzero blocks at $b=5 \times 10^6$, and over 100 times fewer at $b=10^6$. Our SpMM experiments demonstrate this improves the scalability of our approach. 

\minisec{Comparison with Hypergraph partitioning}
HYPE was able to effectively partition the graphs with low maximum degree. On the GAP-twitter and sk-2005 graphs, the results are mixed. Especially for smaller number of partitions ($\leq 32$) we observe a high partition cost. We attribute this to the power-law distribution of those graphs, which makes them challenging to partition in a balanced way. On the MAWI series of graphs, the partitioning is completely ineffective as it leads to one partition being connected to an overwhelming majority of the other vertices. This is because fundamentally, the partitioning cost is lower bounded by the maximum degree, which we overcome using pruning.

\subsection{SpMM Experiments}

\minisec{Strong Scaling}
\begin{figure}
    \centering
    \vspace{-1em}
    \begin{subfigure}[t]{0.48\linewidth}
        \includegraphics[width=\textwidth, trim=0.25cm 0.25cm 0.25cm 0cm, clip]{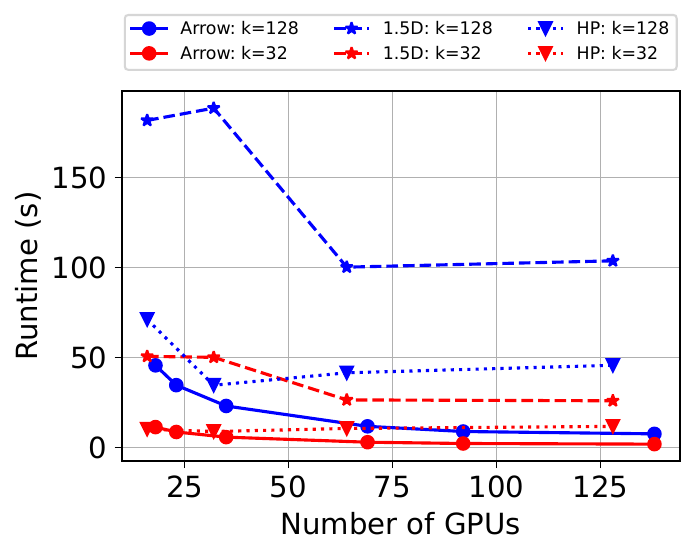} 
        \caption{MAWI 69M}
    \label{fig:mawi-69M} 
    \end{subfigure}    
    \begin{subfigure}[t]{0.48\linewidth}
        \includegraphics[width=\textwidth, trim=0.25cm 0.25cm 0.25cm 0cm, clip]{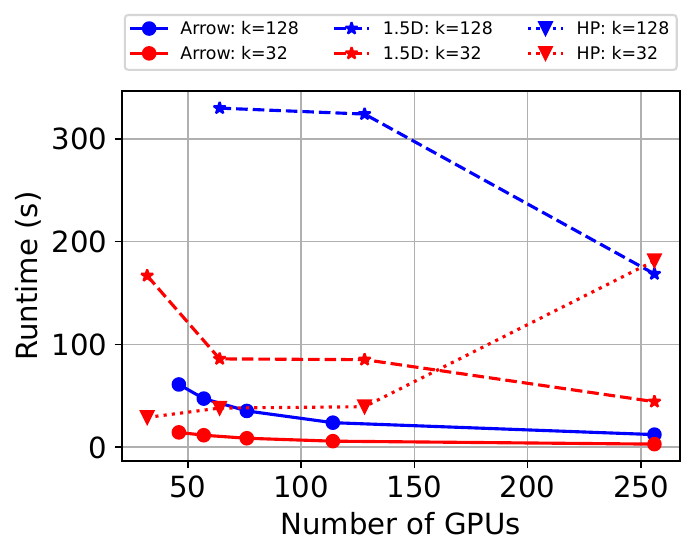}
        \caption{MAWI 226M} 
        \label{fig:mawi-226M} 
        \vspace{0.5em}
    \end{subfigure}
    \begin{subfigure}[t]{0.48\linewidth}
        \centering
        \includegraphics[width=\textwidth, trim=0.25cm 0.25cm 0.25cm 1.5cm, clip]{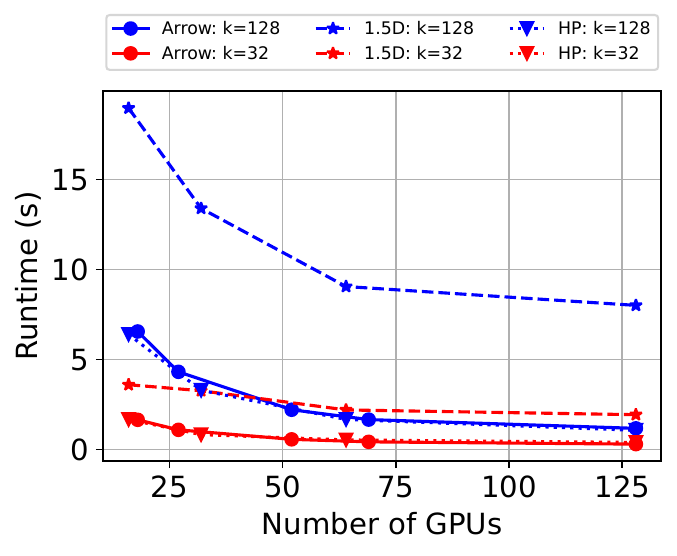}
        \caption{OSM Europe 51M} 
        \label{fig:europe-osm} 
    \end{subfigure}
    \begin{subfigure}[t]{0.48\linewidth}
        \includegraphics[width=\textwidth, trim=0.25cm 0.25cm 0.25cm 1.5cm, clip]{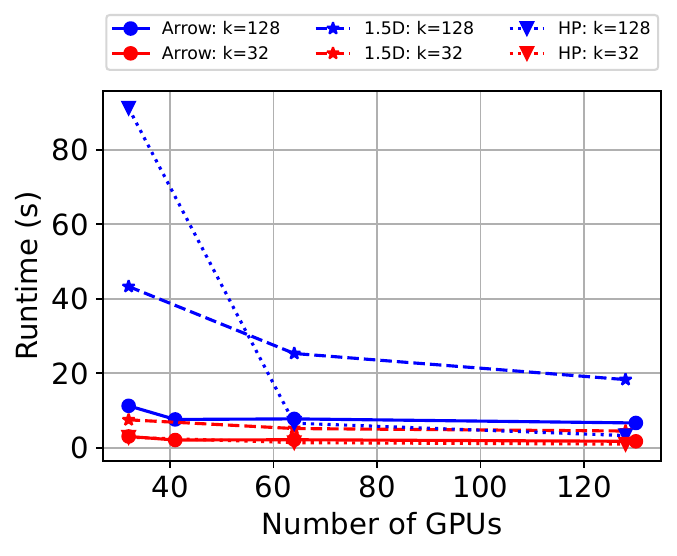}
        \caption{Webbase 118M} 
        \label{fig:webbase} 
        \vspace{0.5em}
    \end{subfigure}
    \begin{subfigure}[t]{0.48\linewidth}
        \includegraphics[width=\textwidth, trim=0.25cm 0.25cm 0.25cm 1.5cm, clip]{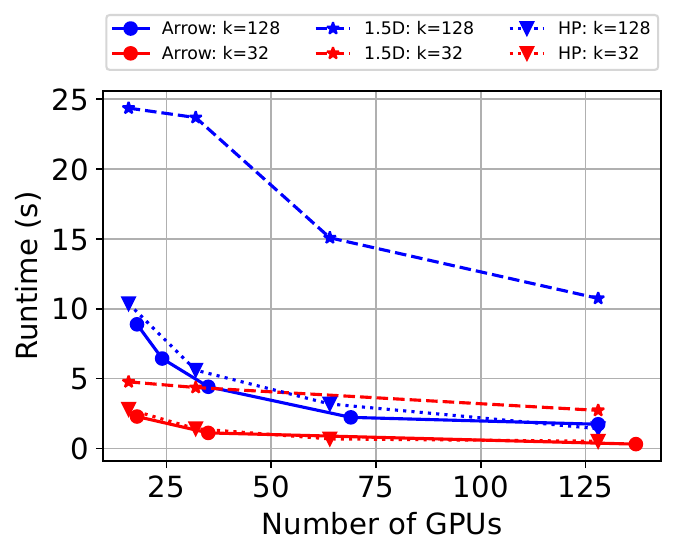}
        \caption{GenBank 68M} 
    \label{fig:genbank-68M} 
    \end{subfigure}
    \begin{subfigure}[t]{0.48\linewidth}
        \includegraphics[width=\textwidth, trim=0.25cm 0.25cm 0.25cm 1.5cm, clip]{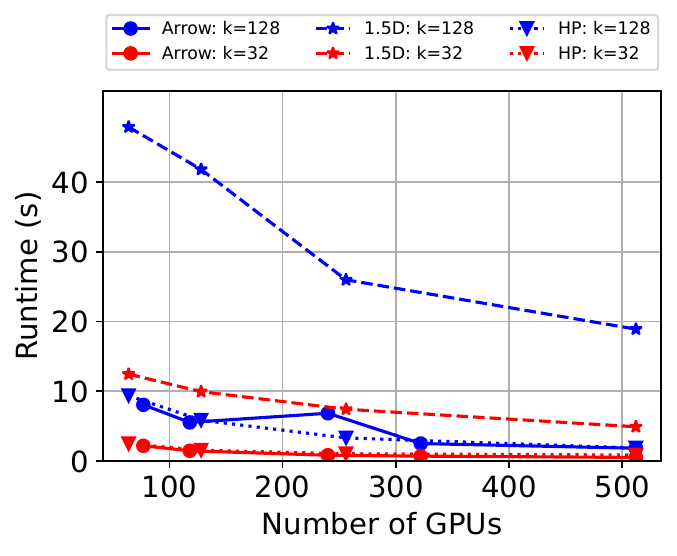}
        \caption{GenBank 214M} 
    \label{fig:genbank-214M} 
    \end{subfigure}
    \begin{subfigure}[t]{0.48\linewidth}
        \includegraphics[width=\textwidth, trim=0.25cm 0.25cm 0.25cm 1.5cm, clip]{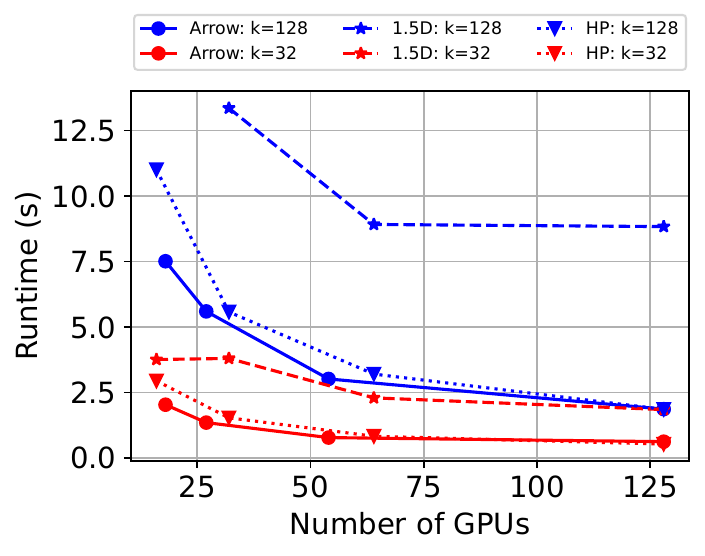}
        \caption{sk-2005 51M} 
    \label{fig:genbank-68M} 
    \end{subfigure}
    \begin{subfigure}[t]{0.48\linewidth}
        \includegraphics[width=\textwidth, trim=0.25cm 0.25cm 0.25cm 1.5cm, clip]{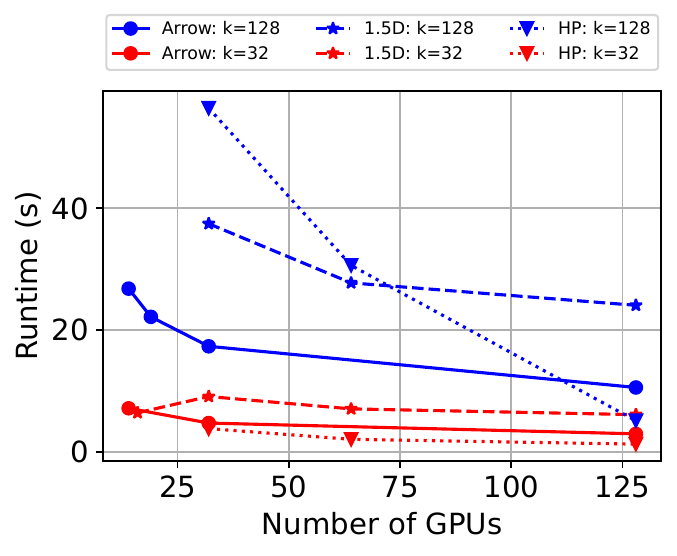}
        \caption{GAP-twitter 62M} 
    \label{fig:genbank-214M} 
    \end{subfigure}
    \vspace{-0.5em}
    \caption{Strong scaling results for varying features sizes.}
    \label{fig:strong-scaling}
    \vspace{-1em}
\end{figure}
Figure \ref{fig:strong-scaling} summarizes our strong scaling evaluation for varying GPU counts and feature matrix column sizes ($k\in {32, 128}$).
Our approach outperforms the 1.5D baseline by significant margins in all instances, except for GAP-twitter on 16 ranks and k=32. In the other cases, the speedup is between 1.7x and 14x.
For the MAWI graphs, the hypergraph partitioning baseline does not scale and is up to $58$ times slower than our approach. On sk-2005 and GAP-twitter our approach is up to 1.4 and 2 times faster, respectively. On the other graphs, our approaches has a similar runtime with HP-1D. 
Generally, the more features, the greater the runtime improvement over both baselines. The more skewed the degree distribution is, the larger our improvement over the hypergraph partitioning baseline.
%

%
%
%
We noticed significant load imbalance in GPU kernels on the MAWI graphs, reaching up to $8$x. 
This imbalance renders both baselines less effective in reducing SpMM kernel runtime with increasing ranks. Moreover, when $c^2 \leq p$, more communication rounds are needed for the 1.5D baseline. 
Hence, sometimes increasing the number of ranks increases the runtime on MAWI. 
We observe that the MAWI datasets contain very large star subgraphs, leading to a high maximum degree. Moreover, they cause load imbalance issues on the GPU, which we think responsible for the comparatively long execution times on the MAWI instances.

\minisec{Weak Scaling}
We applied our decomposition to all MAWI datasets, maintaining a consistent arrow width of $3$ million for a constant computational load per rank. As depicted in Figure \ref{fig:mawi-weak-scaling}, as the dataset size grows from 19 million to 226 million vertices, the runtime only increases marginally by $2.4-6.2\%$. In contrast, the 1.5D baseline decelerates by a factor of $3-3.18$ when going from the smallest to the largest dataset, along with poorer absolute runtimes. The Hypergraph partitioning baseline's runtime grows near-linearly with the number of rows. This is because of the inherent limitations of a 1D decomposition on matrices with highly skewed degree distributions.
\begin{figure}[tb]
    \centering
    \resizebox{\linewidth}{!}{
    \includegraphics[trim=0.5cm 1.1cm 0.5cm 0.1cm]{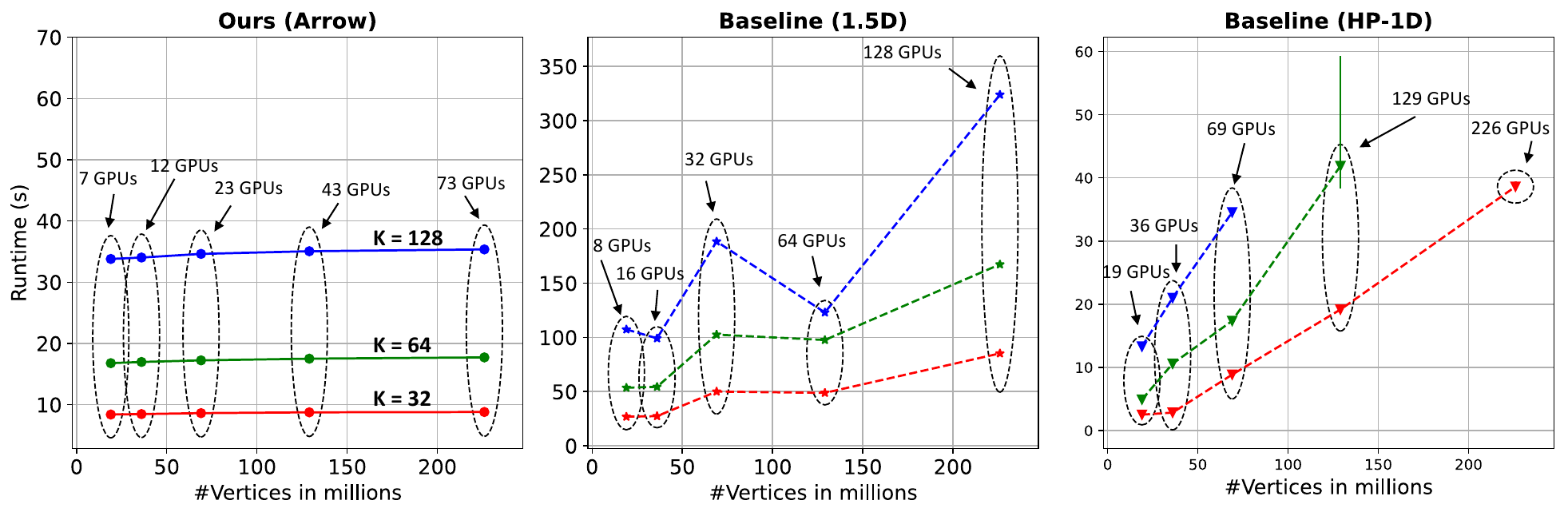}
    }
    \caption{Weak scaling on the MAWI datasets. Missing datapoints represent out-of-memory events.}
    \label{fig:mawi-weak-scaling}
    \vspace{-1.1em}
\end{figure}

\section{Conclusion}
In conclusion, we introduce a novel matrix decomposition, investigating its properties both theoretically and experimentally. The resulting matrices exhibit a distinctive 'arrow' structure that facilitates distributed computation. We demonstrate its efficacy in scenarios of extreme sparsity, where graphs possess an average of $2$ to $38$ nonzeros per row, even in cases of severely skewed degree distribution.

Building upon this matrix decomposition, we propose a distributed approach for sparse times tall-skinny-dense matrix multiplication. This method, grounded in a 1.5D decomposition of the arrow matrices, addresses the memory constraints inherent in direct 1D and 1.5D decompositions. Remarkably, it has low communication costs, all the while avoiding the need to partition the computation into numerous steps. We substantiate this claim through both theoretical analysis and empirical evaluation.

Our comparison against the 1.5D decomposition and a 1D hypergraph partitioning demonstrate our method's scalability and efficiency. Strong scaling experiments highlighted its superiority in handling larger matrices and features, while weak scaling tests showcased its stable runtime with growing datasets.



\begin{acks}
This project received funding from the European Research Council under the European Union's Horizon 2020 programme (Project MAELSTROM, No. 95513) and the Horizon Europe programme (Project GLACIATION, No. 101070141). This work was also supported by the PASC DaCeMI project, as well as UrbanTwin, a joint initiative of the Board of the Swiss Federal Institutes of Technology, in the strategic areas of Energy, Climate and Environmental Sustainability and Engagement and Dialogue with Society. A.N.Z. and T.B.N. are funded by the Swiss National Science Foundation ("QuaTrEx" Project No. 20935 and Ambizione Project No. 185778, respectively). The authors wish to acknowledge the Swiss National Supercomputing Center (CSCS) for providing computing infrastructure.
\end{acks}

\bibliographystyle{acm}

\bibliography{main.bib}

\end{document}